\documentclass[10pt]{article}

\usepackage{lipsum}
\usepackage{amsfonts}
\usepackage{graphicx}
\usepackage{epstopdf}
\usepackage{algorithmic}
\usepackage{dsfont}
\usepackage[utf8]{inputenc} 
\usepackage[T1]{fontenc}
\usepackage{url}
\usepackage{ifthen}
\usepackage{cite}
\usepackage{graphicx,colordvi,psfrag}
\usepackage{amsmath,amssymb}
\usepackage{amsopn}
\usepackage{epstopdf}
\usepackage{epsfig}
\usepackage{calc,pstricks, pgf, xcolor}
\usepackage{bbding}
\usepackage{bbm}
\usepackage{dsfont}
\usepackage{centernot}
\usepackage{bm}
\usepackage[colorlinks=true]{hyperref}
\usepackage[margin=1in]{geometry} 
\usepackage{authblk}
\usepackage[scale=1]{XCharter} 
\usepackage[bigdelims,vvarbb]{newpxmath}

\newcommand{\Ind}{\mathds{1}}

\newcommand{\bR}{\mathbf{R}}

\newcommand{\RR}{\mathbb{R}}
\newcommand{\CC}{\mathbb{C}}

\newcommand{\trace}{\mathop{\mathrm{tr}}}

\newcommand{\Expt}{\mathbb{E}}
\newcommand{\eps}{\varepsilon}

\newcommand{\Unif}{\mathrm{Uniform}}

\newcommand{\m}{\mathcal}

\newcommand{\T}{\top}

\newcommand{\SNR}{{\mathrm{SNR}}}
\newcommand{\Sphere}{{\mathbb{S}}}
\newcommand{\signalClass}{{\mathbb{X}}}
\newcommand{\limp}{\overset{p}{\to}}
\newcommand{\SC}{{n^*_{\text{MRA}}}}
\newcommand{\MSE}{{\mathrm{MSE}^*_{\text{MRA}}}}
\newcommand{\SCp}{{n^*_{\text{PMRA}}}}
\newcommand{\MSEp}{{\mathrm{MSE}^*_{\text{PMRA}}}}

\DeclareMathOperator*{\argmin}{\arg\!\min}
\DeclareMathOperator*{\argmax}{\arg\!\max}

\newif\ifRevComments
\newcommand{\revAdd}[1]{{\ifRevComments\color{teal}\fi#1}}
\newcommand{\revDel}[1]{{\ifRevComments\color{gray}~#1\fi}}

\newtheorem{theorem}{Theorem}[section]

\newtheorem{proposition}[theorem]{Proposition}
\newtheorem{corollary}[theorem]{Corollary}
\newtheorem{remark}[theorem]{Remark}

\newtheorem{lemma}[theorem]{Lemma}

\newenvironment{proof}[1][Proof]{\noindent\textbf{#1.} }{\ \rule{0.5em}{0.5em}}

\numberwithin{equation}{section}

\title{Multi-Reference Alignment in High Dimensions: Sample Complexity and Phase Transition}

\author[1]{Elad Romanov \thanks{E-mail: elad.romanov@gmail.com}}
\author[2]{Tamir Bendory \thanks{E-mail: bendory@tauex.tau.ac.il}}
\author[1]{Or Ordentlich\thanks{E-mail: or.ordentlich@mail.huji.ac.il}}

\affil[1]{School of Computer Science and Engineering, The Hebrew University, 
	Jerusalem, Israel}
\affil[2]{School of Electrical Engineering, Tel Aviv University, Tel Aviv, Israel}
\date{}

\begin{document}
	\maketitle

	
	\begin{abstract}
		Multi-reference alignment entails estimating a signal in $\RR^L$ from its circularly-shifted and noisy copies.
		This problem has been studied thoroughly in recent years, focusing on the finite-dimensional setting (fixed $L$). 
		Motivated by single-particle cryo-electron microscopy, we analyze the  sample complexity of the problem in the high-dimensional regime $L\to\infty$.
		Our analysis uncovers a phase transition phenomenon governed by the parameter $\alpha = L/(\sigma^2\log L)$, where~$\sigma^2$ is the variance of the noise.
		When $\alpha>2$, the impact of the unknown circular shifts on the sample complexity is minor. Namely,  the number of measurements required to achieve a desired accuracy $\varepsilon$ approaches $\sigma^2/\varepsilon$ for small $\varepsilon$; this is the sample complexity of estimating a signal in additive white Gaussian noise, which does not involve shifts.
		In sharp contrast, when $\alpha\leq 2$, the problem is significantly harder and the sample complexity grows substantially quicker with~$\sigma^2$.
	\end{abstract}
	
	
\section{Introduction}

We study the sample complexity of the multi-reference alignment (MRA) model: the problem of estimating a signal from its circularly-shifted and noisy copies.
Specifically, let $X\sim\m{N}(0,I)$ be an $L$-dimensional vector with i.i.d.\ standard normal entries.
We collect  $n$ independent measurements of random cyclic shifts of $X$, corrupted by additive white Gaussian noise: \begin{align} \label{eq:model}
	Y_i=R_{\ell_i}X+\sigma Z_i,  \qquad i=1,\ldots,n,
\end{align} 
where $R_\ell$ denotes a cyclic shift, namely, $(R_\ell X)_{j}=X_{{(j+\ell)}\bmod L}$ for all $j=0,\ldots,L-1$,  
$Z_i\stackrel{i.i.d.}{\sim}\m{N}(0, I)$, and {$\ell_i\stackrel{i.i.d.}{\sim}\Unif(\{0,\ldots,L-1\})$} are statistically independent of $X$. 
Given the measurements $Y^n=(Y_1,\ldots,Y_n)$, one is interested in constructing an estimator $\hat{X}=\hat{X}(Y^n)$ of the signal. Importantly,  the unknown shifts $\ell_1,\ldots,\ell_n$---while their estimation might be a means to an end---are nuisance variables. 
\revAdd{Figure~\ref{fig:example} shows an example of a measurement drawn from~\eqref{eq:model}.}

\begin{figure}[ht]
	\centering
	\includegraphics[width=.8\columnwidth]{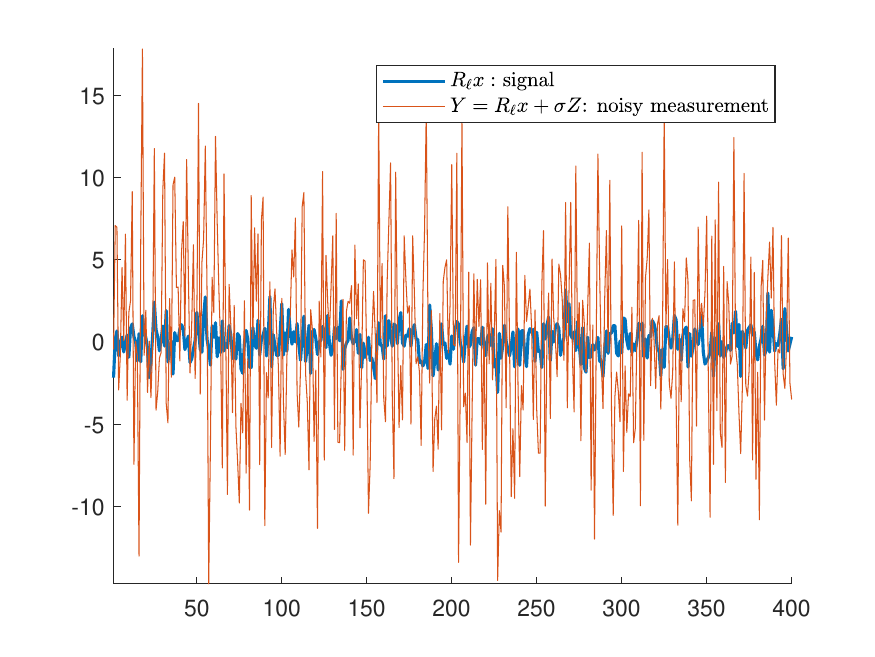}
	\caption{\label{fig:example}  \revAdd{An example of a measurement drawn from~\eqref{eq:model} for $\alpha=2$ and $L=400$. The corresponding noise level is $\sigma^2 = 33.38$.}}	
\end{figure}

This paper focuses on the high-dimensional regime, where the dimension of the signal grows indefinitely $L\to\infty$.
In this setting, we wish to characterize the relations between the number of measurements $n$, the length of each observation $L$, and  the noise level $\sigma^2$ that allow estimating~$X$ to a prescribed accuracy. This is in contrast to previous works, surveyed in Section~\ref{sec:prior_art}, which analyzed the interplay between~$n$ and $\sigma$, while considering a fixed $L$.

It is important to note that given the measurements, there is no way to distinguish between $X$ and its cyclic shift since $P_{Y^n|X=x}=P_{Y^n|X=R_1 x}=\cdots= P_{Y^n|X=R_{L-1}x}$. 
Therefore, we can only estimate the orbit of $X$ under the group of circular shifts $\mathbb{Z}_L$. Accordingly, we use the following  distortion measure
\begin{align} \label{eq:rho_def}
	\rho(X,\hat{X})=\frac{1}{L}\min_{\ell=0,\ldots L-1}\|X-R_{\ell}\hat{X}\|^2.
\end{align}
In the sequel, we loosely {say that we aim to estimate $X$ rather than its orbit, and refer to  $\mathbb{E}\rho(X,\hat{X})$ as the MSE.}

\paragraph{Sample complexity} Our goal in this paper is to characterize  the smallest possible number of measurements required to achieve a desired  MSE  in terms of
{the dimension $L$ and the noise level~$\sigma^2$.}
To that end, we define the smallest MSE attainable by any estimator as
\begin{align}
	\MSE(L,\sigma^2,n):=\inf_{\hat{X}}\mathbb{E}\rho(X,\hat{X}(Y^n)),
\end{align}
and the sample complexity of the MRA problem 
\begin{align}
	\SC(L,\sigma^2,\varepsilon):= \min\left\{n \ :   \MSE(L,\sigma^2,n)\leq \varepsilon  \right\}.
\end{align}

We define the signal-to-noise ratio (SNR) by 
\begin{equation}
	\text{SNR} := \frac{\Expt\|X\|^2}{\sigma^2}=\frac{L}{\sigma^2}.
\end{equation}
This definition is consistent with previous works which considered a fixed $L$ and $\sigma\to\infty$, implying SNR$\to 0$; see~Section~\ref{sec:prior_art}.

The asymptotics in our model turn out to be particularly interesting when the dimension, {the noise level}, and the SNR are simultaneously large. In particular,  it will be convenient to parametrize the noise variance by 
\begin{align} \label{eq:alpha}
	\sigma^2(\alpha)=\frac{L}{\alpha \log{L}} \quad \Longleftrightarrow \quad  \alpha=\frac{L}{\sigma^2\log{L}} = \frac{\text{SNR}}{\log L}.
\end{align}
Accordingly, we define $\MSE(L,\alpha,n):=\MSE(L,\sigma^2(\alpha),n)$ and $\SC(L,\alpha,\varepsilon):=\SC(L,\sigma^2(\alpha),\varepsilon)$.

\paragraph{Motivation}
The MRA model is mainly motivated by single-particle cryo-electron microscopy (cryo-EM)---a leading technology to constitute the 3-D structure of biological molecules. 
In its most simplified  version, the cryo-EM problem involves reconstructing a 3-D structure from its multiple noisy tomographic projections, taken after the structure has been rotated by an unknown 3-D {rotation}. 
In analogy, in the MRA model~\eqref{eq:model} the signal $X$ is measured after an unknown circular shift. 
In Theorem~\ref{th:PMRA}, we  extend the basic model to include a projection; we refer to this model as the projected MRA model. This projection plays the role, to some extent, of the tomographic projection in cryo-EM.  {Section~\ref{sec:conclusion}  discusses further potential extensions.}

The correspondence between MRA and cryo-EM, while admittedly not perfect, {has motivated an extensive  study of the MRA problem in recent years.} 
For example, the resolution limitations of MRA were analyzed in~\cite{bendory2020super} in order to draw an analogy to the achievable resolution of cryo-EM---a crucial aspect from a biological standpoint.  More relevant to this work, in~\cite{bandeira2017estimation,perry2019sample,bandeira2020non,abbe2018estimation}, the sample complexity of the MRA and cryo-EM models were analyzed for a fixed dimension~$L$. 
Remarkably, it was shown that in the low noise regime  (small $\sigma$), the number of measurements should scale like $\sigma^2$, while  in the high noise regime (large $\sigma$) $n$ must increase with~$\sigma^6$; see further discussion in Section~\ref{sec:prior_art}.

Our high-dimensional analysis is motivated by the size of modern cryo-EM datasets. In a  typical cryo-EM experiment, the number of measurements and the dimension of the 3-D structure are  of the same order of a few millions. For example, a 3-D structure of size $200\times 200\times 200$ voxels resulting in $8,000,000$ parameters to be estimated.
\revAdd{Since a typical noise level in a cryo-EM dataset is $\sigma^2\approx 100$,  the anticipated parameter regime is  $\alpha\gg 1$.} \revAdd{We do emphasize, however, that these numbers should be taken with some degree of skepticism: while cryo-EM is a motivation for studying the MRA problem, these are ultimately quite different problems, and practical cryo-EM setups involve additional complications, that are not captured by MRA \cite{bendory2020single}.}
In fact, high-dimensional statistical analysis has been already proven to be effective for cryo-EM data processing. 
For example, a covariance estimation technique based on high-dimensional analysis (the so-called spiked model) has significantly improved image denoising~\cite{bhamre2016denoising}.

\paragraph{Information-theoretic background and asymptotic notation} The analysis of this work is greatly based on information-theoretic notions and techniques. For completeness, we review the relevant definitions in supporting information (SI) appendix, Section~\ref{sec:ITback}.

We also repeatedly use asymptotic notation.
For sequences $a=a(L)$ and $b=b(L)$, we write $a(L)=O(b(L))$ if there exists a constant $C>0$ such that $a(L)\le Cb(L)$ for all $L$. Similarly, $a(L)=\Omega(b(L))$ means $a(L)\ge Cb(L)$. Occasionally, we use $a(L)=O_\beta(b(L))$ to signify explicitly that $C$ depends on some parameter $\beta$. The notation $a(L)=o(b(L))$ means $a(L)/b(L)\to 0$ as $L\to \infty$. In particular, if $a(L)=o(1)$ then $a(L)\to 0$ asymptotically.
Similarly, $a(L)=\omega(b(L))$ means $a(L)/b(L)\to\infty$.  

\paragraph{Reproducibility}
The code to reproduce the figures is publicly available at~\url{https://github.com/TamirBendory/high-dimensional-mra-bounds}.\footnote{Our expectation-maximization implementation is based on the code of~\cite{bendory2017bispectrum}.}

\paragraph{Supporting information (SI)} Due to space constraints, we have relegated the proofs of several technical claims to the SI appendix. In addition to those, the SI contains a brief review of all information-theoretic notions necessary to follow this work (Section~\ref{sec:ITback}), as well as some additional discussion which is somewhat tangential to our main results (Section~\ref{sec:capacity}).

\section{Main results and discussion} \label{sec:main_results}

\paragraph{Phase transition.} This work focuses on the asymptotic setting where $L$ tends to infinity. 
Our first main finding is that in this asymptotic limit there is a transition in terms of the behavior of the sample complexity. For $\alpha>2$, 
the MRA problem is essentially as easy as estimating a signal in additive white Gaussian noise (AWGN), with no random shifts. More precisely, for sufficiently small distortion $\varepsilon$, the sample complexity tends to the sample complexity of estimating a signal in AWGN,  ${n^*_{\text{AWGN}}(L,\alpha,\varepsilon)}=\lceil\left(\frac{1}{\varepsilon}-1\right)\sigma^2(\alpha) \rceil$, which behaves as $\frac{\sigma^2(\alpha)}{\varepsilon}$ for small $\varepsilon$.
In sharp contrast, for  $\alpha\le 2$  the problem becomes substantially harder. 

\begin{theorem}The sample complexity of the MRA model~\eqref{eq:model}  obeys:
	\begin{enumerate}
		\item \label{it:highsnr}For any $\alpha>2$ we have
		\begin{align*}
			\lim_{\varepsilon\to 0}\lim_{L\to\infty}\frac{\SC(L,\alpha,\varepsilon)}{\sigma^2(\alpha)/\varepsilon}=\lim_{\varepsilon\to 0}\lim_{L\to\infty}\frac{\SC(L,\alpha,\varepsilon)}{n^*_{\text{AWGN}}(L,\alpha,\varepsilon)}=1.
		\end{align*}
		\revAdd{
			\item \label{it:lowsnr}For any $\alpha\leq 2$ and any $\varepsilon<1$ we have
			\begin{align*}
				\SC(L,\alpha,\varepsilon) = \omega\left( \sigma^2 \log\left(1/\varepsilon\right) \right) \,.
			\end{align*}
			In particular, for fixed $\varepsilon$,
			\begin{align*}
				\lim_{L\to\infty}\frac{\SC(L,\alpha,\varepsilon)}{n^*_{\text{AWGN}}(L,\alpha,\varepsilon)}=\infty.
			\end{align*}
		}
	\end{enumerate}
	\label{thm:highsnrregime}
\end{theorem}

In part~\ref{it:highsnr} of Theorem~\ref{thm:highsnrregime}, the lower bound $\frac{\SC(L,\alpha,\varepsilon)}{n^*_{\text{AWGN}}(L,\alpha,\varepsilon)}\geq 1$ is trivial: estimating in the MRA model is harder than estimating a signal in AWGN (namely, when the shifts are known). 
A small subtlety is that the distortion measure $\mathbb{E}\rho(X,\hat{X})$ is a bit weaker than the standard definition of MSE, $\mathbb{E}\|X-\hat{X}\|^2$, as it allows for any cyclic shift. However, we show in Section~\ref{sect:ITbounds} that, as expected, this has a vanishing effect for large $L$. 
In order to show that $\lim_{\varepsilon\to 0}\lim_{L\to\infty}\frac{\SC(L,\alpha,\varepsilon)}{n^*_{\text{AWGN}}(L,\alpha,\varepsilon)}\leq 1$ we introduce  an algorithm that for any $\alpha>2$ requires about $\sigma^2(\alpha)/\varepsilon$ samples to achieve $\mathbb{E}\rho(X,\hat{X})\leq \varepsilon$, provided that $\varepsilon$ is sufficiently small and $L$ is sufficiently large. 
The sole purpose of the estimation procedure is establishing an upper bound; its computational complexity is exponential in $L$ and thus the procedure is far from being efficient.  More specifically, it is based on a two-step procedure. First, we construct a $\delta$-net that, by definition, contains a member close to $X$ and look for the most likely candidate within that net given the measurements. Second, we use this candidate in order to determine almost all shifts $\hat\ell_i$, and then estimate the signal by alignment and averaging  $\hat X=\frac{1}{n}\sum_{i=1}^nR_{-\hat \ell_i}Y_i$.
The details are given in Section~\ref{sec:upper-bound}. 

In order to establish part~\ref{it:lowsnr} of Theorem~\ref{thm:highsnrregime}, we show that for $\alpha\leq 2$ the mutual information (MI) $I(X;Y)$  between $X$ and a single MRA measurement  grows with $L$ significantly slower than {$I(X;X+\sigma Z)$, as} in estimating a signal in AWGN. The details are given in  Section~\ref{sect:ITbounds}.

\revAdd{
	Although our results are \emph{asymptotic} in $L$,  the transition in the difficulty of the problem around $\alpha=2$, as predicted by Theorem~\ref{thm:highsnrregime}, is evident already for relatively small $L$. 
	Figure~\ref{fig:em} presents the root MSE (RMSE) as a function of $\alpha$ for different values of $L$. 
	We take our estimator $\hat{X}$ to be the output of the expectation-maximization (EM) algorithm~\cite{dempster1977maximum,bendory2017bispectrum}, which is the {standard} choice for {MRA}; see details in Section~\ref{sec:prior_art}.
	For large values of $L$ and large $\alpha$, the error of EM tends to that of estimating a signal in AWGN, implying that it detects the shifts accurately. For smaller values of $\alpha$, the error grows rapidly, especially when $\alpha<2$.
	We note that the observed transition in the vicinity of $\alpha=2$, at the values of $L$ considered in Figure~\ref{fig:em} (few $100$s), appears to not be very sharp. Our proofs suggest that perhaps this behavior is to be expected: the concentration rates we are able to derive for some of the quantities relevant to the analysis is quite slow (inverse polynomial in $L$, with a very small exponent when $\alpha$ is close to $2$).
} 
\revDel{
	
	Although our results are asymptotic in $L$,  the phase transition at $\alpha=2$ predicted by Theorem~\ref{thm:highsnrregime} is evident already for relatively small $L$. 
	Figure~\ref{fig:em} presents the root MSE (RMSE) as a function of $\alpha$ for different values of $L$. 
	We take our estimator $\hat{X}$ to be the output of the expectation-maximization (EM) algorithm~\cite{dempster1977maximum,bendory2017bispectrum}, which is the {standard} choice for {MRA}; see details in Section~\ref{sec:prior_art}.
	For large values of $L$ and large $\alpha$, the error of EM tends to that of estimating a signal in AWGN, implying that it detects the shifts accurately. For smaller values of $\alpha$, the error grows rapidly, especially when $\alpha<2$. 
}

\begin{figure}[ht]
	\centering
	\includegraphics[width=.8\columnwidth]{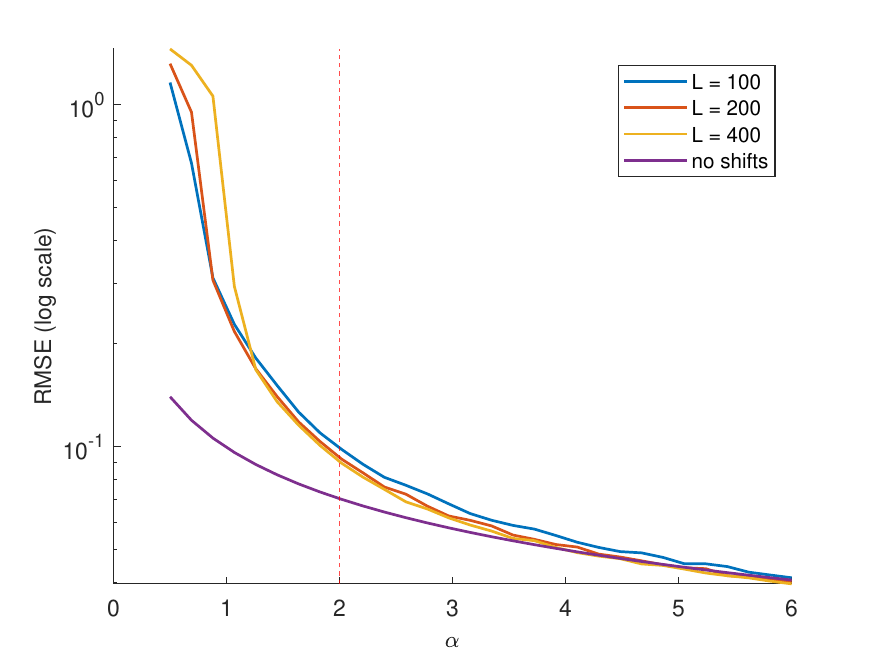}
	\caption{\label{fig:em}  The RMSE of EM (averaged over 100 trials) as a function of $\alpha$ for different values of~$L$. The number of measurements was set to be $n(L) = 100L/\log(L)$. \revAdd{An example of a single measurement appears in Figure~\ref{fig:example}.}
		For large values of~$\alpha$, the error reduces to the error of estimating a signal in AWGN,  $\sqrt{\frac{\sigma^2}{\sigma^2+n}}=\frac{1}{\sqrt{1+100\alpha}},$  suggesting that EM performs as if the shifts were known. 
		For small values of $\alpha$, and in particular $\alpha<2$, the error rapidly increases.}	
\end{figure}


\paragraph{Connection with template matching}  
At this point, the reader may wonder what is the intuitive interpretation of $\alpha=2$. To answer this question we now  introduce the \emph{template matching problem}, which is studied in detail in Section~\ref{sect:template-matching}. In this problem, we are given $X$ and one MRA measurement $Y=R_\ell X+Z$, where $X$, $R_\ell$ and $Z$ are distributed as above, and our goal is to recover the shift $R_{\ell}$. 
We will see that in the asymptotic setting, $\alpha=2$ is the critical threshold for this problem. That is, the error probability in recovering $R_\ell$ from $(X,Y)$ approaches $0$ for all $\alpha>2$, and approaches $1$ for all $\alpha<2$.  

In the MRA problem, recovering the shifts is harder, as we do not have access to $X$. We nevertheless show that for $\alpha>2$, given enough measurements, it is possible to recover a fraction approaching~$1$ of the shifts correctly. 
On the other hand, recovering a large fraction of the shifts correctly for $\alpha< 2$  is impossible since it is impossible even in the template matching model.
Intuitively, if we cannot recover almost all shifts, the attained MSE should be much worse than in estimating a signal in AWGN, which means that the sample complexity should be much higher for $\alpha<2$.  Our bounds in Section~\ref{sect:ITbounds} formalize this intuition.

To illustrate the phase transition for template matching, we conducted a ``genie-aided'' experiment, presented in Figure~\ref{fig:genie}.
In this experiment, we use the true $X$ (the ``genie'') in order to estimate the shifts by $\hat\ell_i=\arg\max_{\ell\in\{0,\ldots,L-1\}} \langle R_\ell X,Y_i\rangle $. Then, we  estimate the signal by aligning the measurements and averaging  $\hat{X}=\frac{1}{n} \sum_{i=1}^nR_{-\hat{\ell_i}}Y_i$. 
For large values of $\alpha$, the recovery error converges to the  error of estimating a signal in AWGN. For smaller $\alpha$ values, and in particular $\alpha<2$, the recovery error rapidly increases.

\begin{figure}[ht]
	\centering
	\includegraphics[width=0.8\columnwidth]{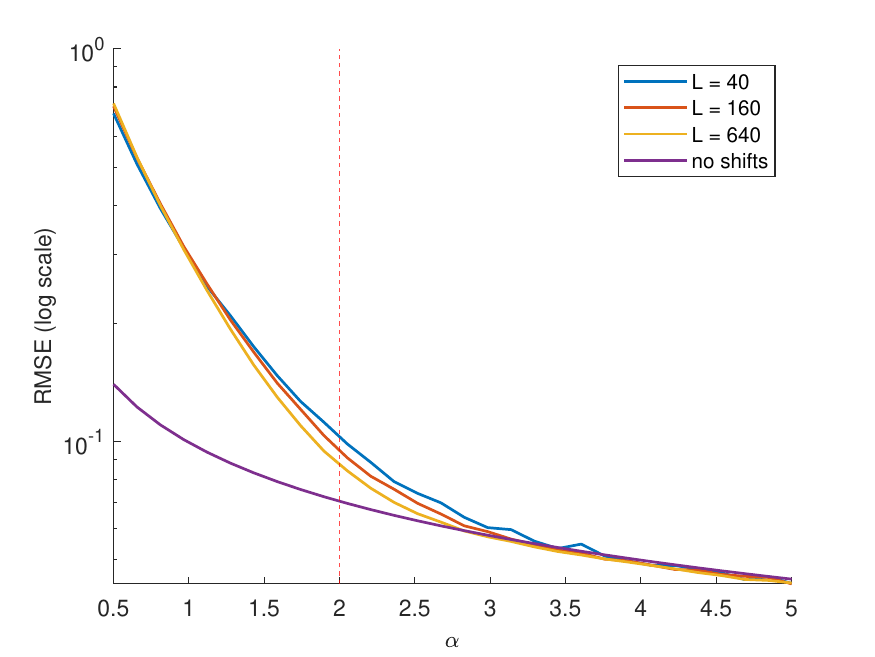}
	\caption{\label{fig:genie} A ``genie-aided''  experiment:
		the true $X$ is used to estimate the shifts $\hat\ell_1,\ldots,\hat\ell_n$, as in the template matching problem, and then the signal is estimated by aligning all measurements and averaging $\hat{X}=\frac{1}{n}\sum_{i=1}^nR_{-\hat{\ell_i}}Y_i$. The figure presents the RMSE (averaged over 50 trials) as a function of $\alpha$ for different values of $L$. The number of measurements was set to be $n(L) = 100L/\log(L)$. For large values of $\alpha$, the error reduces to the error of estimating a signal in AWGN (i.e., when the shifts are known) $\sqrt{\frac{\sigma^2}{\sigma^2+n}}=\frac{1}{\sqrt{1+100\alpha}}$. 
		For small values of $\alpha$, and in particular $\alpha<2$, the template matching error quickly increases.
	} 
\end{figure}

\paragraph{Tighter lower bound for the low SNR regime}
Theorem~\ref{thm:highsnrregime} shows that for all $\alpha\leq 2$ and fixed ${\varepsilon<1}$ the shifts make a difference: 
the sample complexity with unknown shifts (i.e., the MRA problem) is \revAdd{$\omega\left(\sigma^2(\alpha)\log(1/\varepsilon)\right)$}, and is therefore substantially greater than the sample complexity when the shifts are known.
For $\alpha<1$, we were able to prove a much stronger lower bound on the sample complexity.

\begin{theorem}
	\revAdd{
		For any $0<\alpha<1$, and $0<\eps<1$,
		\begin{align}
			\SC(L,\alpha,\varepsilon)= \Omega\left(  L^{2-\alpha}\log(1/\eps) \right) \,.
		\end{align}
	}
	\label{thm:alphalessthan1}
\end{theorem}
Theorems~\ref{thm:highsnrregime} and \ref{thm:alphalessthan1} are proved in Section~\ref{sect:ITbounds}. 

\paragraph{The sample complexity of the projected MRA model} 
Recall that MRA serves as a toy model of the cryo-EM reconstruction problem. An additional complication arising in cryo-EM is a fixed tomographic projection, a line integral,  also known as the X-ray transform.
To account for this effect, we extend our basic model~\eqref{eq:model}  to the \emph{projected multi-reference alignment problem} (PMRA) model:\footnote{We mention that other projected MRA models were studied in~\cite{bandeira2017estimation,bendory2020super}.}
\begin{align} \label{eq:PMRA_model}
	Y_i=\pi_S R_{\ell_i}X+\sigma Z_i.
\end{align} 
Here, $\pi_S:\RR^L\to \RR^{L'}$ is matrix projecting a vector in $\RR^L$ to $\RR^{L'}$ by keeping only the coordinates that belong to a subset $S\subset[L]$  of size $L'\leq L$ and discarding the rest{, and $Z_i\stackrel{i.i.d.}{\sim}\m{N}(0,I)$ are $L'$-dimensional i.i.d. Gaussian vectors.}
We assume that $S$ is fixed and known 
to the estimator. 
As in MRA without the projection, the goal is to reconstruct $X$ up to {a circular shift}, that is, produce an estimate $\widehat{X}$ such that $\Expt\rho(X,\widehat{X})$ is as small as possible. 

We study the PMRA problem in an asymptotic setting where $L,L',\sigma^2\to \infty$ simultaneously. It makes sense to adopt a slightly different scaling for the noise in PMRA, as 
\begin{equation}\label{eq:sigma-pmra}
	\sigma^2 = \sigma^2_{\mathrm{PMRA}}(\alpha) = \frac{L'}{\alpha\log(L)} .
\end{equation}
The reason for this particular scaling will be made {clear} from the {analysis}: the numerator is the total signal energy available in a single measurement, $\Expt\|\pi_S R_{\ell_i}X\|^2 = L'$; the $\log(L)$ factor is log the size of the group of shifts. {In Section~\ref{sec:conclusion} we provide some remarks as to how to extend our results to other groups.}
Similarly to our notation for the MRA model, we denote the smallest attainable MSE in the PMRA model as $\MSEp(L,\alpha,n)$, and the sample complexity as $\SCp(L,\alpha,\varepsilon)$. 

\begin{theorem} \label{th:PMRA}
	Suppose that {$\sigma^2_{\mathrm{PMRA}}(\alpha)$} is scaled as in~\eqref{eq:sigma-pmra}, and $L,L'\to\infty$, so that $L'\le L$ and $L'=\omega(\log(L))$ (that is, $L$ grows strictly less than exponentially fast in $L'$).
	The sample complexity of the PMRA model~\eqref{eq:PMRA_model} obeys the following lower bounds:
	\begin{enumerate}
		\item For any $\alpha>2$ and $0<\varepsilon<1$ we have that
		\begin{align}
			\SCp(L,\alpha,\varepsilon)\geq \frac{L}{L'}\left(\frac{1}{\varepsilon}-1\right)\sigma_{\mathrm{PMRA}}^2(\alpha)(1+o(1)) .
		\end{align}
		
		\item For any $\alpha\leq 2$ and $0<\varepsilon<1$ we have that
		\revAdd{
			\begin{align}
				\SCp(L,\alpha,\varepsilon)=\omega\left(\frac{L}{L'} \sigma_{\mathrm{PMRA}}^2(\alpha) \log(1/\varepsilon)\right).
			\end{align}
		}
	\end{enumerate}
	\label{thm:pmra}
\end{theorem}
The proof of the theorem {relies heavily on the proof of Theorem~\ref{thm:highsnrregime}. Due to space constraints, a proof sketch is relegated to the SI appendix, see Section~\ref{sect:pmra-proof}.}
We conjecture that at high SNR ($\alpha>2$), the lower bound given in Theorem~\ref{thm:pmra} is in fact tight at very low MSE (formally $\varepsilon\to 0$, as in Theorem~\ref{thm:highsnrregime}). 


\revAdd{
	
	\paragraph{Extension to other signal priors and group actions}	
	In section~\ref{sec:conclusion} we describe briefly how one could modify our proofs to account for other i.i.d. signal priors (besides Gaussian) and finite group actions. 
}

\section{Prior art} \label{sec:prior_art}
The multi-reference alignment problem was introduced by~\cite{bandeira2014multireference}, and fully formulated in~\cite{bandeira2020non}.  The general MRA model reads 
\begin{align} \label{eq:mra_general}
	Y_i=T_i (g_i \circ X)+\sigma Z_i,  \qquad i=1,\ldots,n, 
\end{align} 
where $g_i$ is a random element of a compact group $G$ (drawn from a possibly unknown distribution over $G$) acting on a vector space $X\in\mathbb{X}$, and $T_i, \, i=1,\ldots,n,$ are known linear operators. If $T_i=I$ for all $i$, $g_i$ are drawn uniformly from the group of cyclic shifts $\mathbb{Z}_L$, and  $X\sim\mathcal{N}(0,I)$, then~\eqref{eq:mra_general} reduces to the MRA model~\eqref{eq:model}.
This model can be thought of as a special case of a Gaussian mixture model, where all centers are connected through a group action (i.e., a cyclic shift).
If $T_i=\pi_S$ for all $i$, we get the projected MRA model~\eqref{eq:PMRA_model}. 
In cryo-EM---the main motivation of this work---$G$ is the group of 3-D rotations $SO(3)$, $\mathbb{X}$ is the space of 3-D ``band-limited'' functions (that is, functions that can be expanded by finitely many basis functions), and $T_i$ encodes the (fixed) tomographic projection, as well as other linear effects, such as the microscope's point spread function (which varies across images) and sampling~\cite{singer2018mathematics,bendory2020single}.

The sample complexity of the MRA model~\eqref{eq:model}, in the minimax sense, was first studied in~\cite{bandeira2017optimal,perry2019sample}.
The focus of these works, as well as the rest of the works mentioned in this section, is  on the regime where the noise level $\sigma$ and number of measurements $n$ diverge, while the dimension of each measurement $L$ is fixed, implying $\SNR\to 0$.  
These results were extended to the general MRA model~\eqref{eq:mra_general}  by~\cite{bandeira2017estimation}  and~\cite{abbe2018estimation} (the latter generalizes the framework proposed in~\cite{abbe2018multireference}). These  papers constitute an intimate connection between the MRA model and the method of moments---a classical estimation technique. 
Let $\bar{d}$ be the lowest order moment that distinguishes two different signals (signals that are not in the same orbit) given a specific MRA model (namely, fixed $T_i,\mathbb{X}$, and a distribution over $G$).  Then, unless {$n\cdot \SNR^{\bar{d}}\to\infty$,} 
the MSE is bounded from below.  More informally, the moments determine the optimal (minimax) estimation rate of the problem. 
For example, for the MRA model~\eqref{eq:model} it is known that the third moment determines a generic signal uniquely (in this work we only consider normal i.i.d.\ signals that fall into this category), i.e., $\bar{d}=3$, and  thus~$n\cdot\SNR^3\gg 1$ 
is a necessary condition.
Remarkably, this phenomenon was observed empirically  in context of cryo-EM  early on by Sigworth~\cite{sigworth1998maximum}.

\revAdd{In this work, we propose an alternative explanation for the statistical difficulty of MRA at low SNR, in a setting where the signal $X$ is ``generic'' (specifically, $X\sim \m{N}(0,I)$) and the dimension is very large. The separation between the two SNR regimes we identify is \emph{not} given in terms of moments;
	instead, it is characterized in terms of a very natural estimation-theoretic question: is it possible, in an information-theoretic sense, to consistently recover the unknown shifts (nuisance parameters) themselves? As we scale $\SNR=\alpha\log L$, the threshold $\alpha=2$, separating the high and low SNR regimes, is exactly the threshold for the shift recovery problem.
	Note that in this high-dimensional setting, we find that the low SNR regime in fact extends beyond the case $\SNR\to 0$ to unbounded values of $\SNR$ (provided that it grows slowly enough with $L$)---this is in contrast to previous works that study MRA in fixed dimension.}
\revDel{In this work we discover that the statistical properties of MRA in high-dimensions, at least for $X\sim\mathcal{N}(0,I)$, are not characterized by moments, but rather by the parameter~$\alpha$ that balances the noise level and the dimension~\eqref{eq:alpha}. 
	In particular, in our setting $\SNR=\alpha\log L$ diverges, rather than $\SNR\to 0$ as in previous works. In this sense, our results imply that the ``low SNR'' regime is not only $\SNR\to 0$, and actually extends into unbounded values of $\SNR$ provided that it grows slowly enough with $L$.}


From the algorithmic perspective, two main computational frameworks were applied to MRA problems. The first approach is based on expectation-maximization (EM)---a popular heuristic to maximize the posterior distribution~\cite{dempster1977maximum}.
EM is  the most popular and successful methodology to elucidate high-resolution 3-D structures using cryo-EM~\cite{scheres2012relion,bendory2020single}, and it was  successfully applied  to a variety of MRA setups~\cite{bendory2017bispectrum,boumal2018heterogeneous,abbe2018multireference,ma2019heterogeneous,bendory2020super}.
A  recent work \cite{fan2020likelihood} studies the likelihood landscape for the general MRA model~\eqref{eq:mra_general}, where $G$ is a discrete group and $T_i=I$. The latter paper shows that when the dimension is fixed and the SNR is sufficiently high, the log likelihood has certain favorable features from an optimization perspective; their results give a compelling argument for why EM seems to give good performance for MRA in high SNR.
In~\cite{brunel2019learning}, it is shown that usually maximum likelihood achieves the parametric rate $\rho(X,\widehat{X}_{\mathrm{MLE}})\sim 1/n$, although in some cases the rate can be $\sim 1/\sqrt{n}$.

The second  algorithmic framework is based on the method of moments. This approach has an appealing property: it requires only one pass over the measurements, and thus its computational load is relatively low, unless $L$ is large~\cite{bendory2017bispectrum,boumal2018heterogeneous,abbe2018multireference,ma2019heterogeneous,perry2019sample,pumir2019generalized}. 
In addition, as mentioned,   it achieves  the optimal estimation rate  when $L$ is fixed and $\SNR\to 0$.  Consequently, a variety of moment-based algorithms were proposed. 
For example, the authors of \cite{perry2019sample} suggest estimating the  third-order  tensor moment of the signal $T^{(3)} = L^{-1}\sum_{\ell=0}^{L-1} (R_{\ell} X)^{\otimes 3}$, from which $X$ can be recovered by Jenrich's method \cite{harshman1970foundations,leurgans1993decomposition}. Using the robustness analysis of \cite{goyal2014fourier}, they were able to show that $n=O\left(\varepsilon^{-1}\sigma^{6}\mathrm{poly}(L)\right)$ samples suffice to achieve $\rho(X,\widehat{X})\le \varepsilon$ with constant probability. This bound depends polynomially on both  the dimensional and on the inverse smallest DFT coefficient of $X$; 
when $X\sim\m{N}(0,I)$, one can verify that typically {all the DFT coefficients of $X$ are greater than $\Omega(L^{-1/2})$.} 
The $\mathrm{poly}(L)$ dependence is not computed explicitly, but to the best of our understanding, the analysis of~\cite{goyal2014fourier} provides a significantly worse dimensional scaling than the $\Omega(L^2)$ in our lower bound (as $\alpha\to 0$).
Another work \cite{bendory2017bispectrum} studies recovery by bispectrum inversion, which is equivalent to the third-order moment if the distribution of shifts is uniform.
They argue that when $L$ is fixed, the sample complexity should scale like $O(\sigma^{6})$,  hiding an implicit dependence on $L$. 
The method of moments was also applied to cryo-EM and related technologies, see for example~\cite{kam1980reconstruction,donatelli2015iterative,levin20183d,sharon2020method}, as well as to additional MRA setups~\cite{abbe2017sample,aizenbud2019rank,hirn2019wavelet}.

A recent work~\cite{katsevich2020likelihood} establishes an enticing connection between likelihood-based techniques and the method of moments for the general MRA model~\eqref{eq:mra_general}  for fixed $L$, $\SNR\to 0$, and $T_i=I$. Specifically, it was shown that  likelihood optimization in the low SNR regime reduces to a sequence of moment matching problems.
In addition, the method of moments is also closely-related to invariant theory and thus tools from the latter field can be applied to analyze MRA models; see in particular~\cite{bandeira2017estimation}.

\section{Phase transition of template matching}\label{sect:template-matching}
Suppose that the shifts $R_{\ell_i}$ are all known. In this scenario, estimating the signal is easy: one needs to align each observation $R_{\ell_i}^{-1}y_i$ and average out the noise. Therefore, if possible, it makes sense to try and estimate the shifts.
In this section, we study the problem of estimating a shift 
when the signal is assumed to be known (which is not the case in MRA); we refer to this problem as \emph{template matching}.
Specifically,  suppose that one has access to a signal, a ``template'' $X\in \RR^L$, and observes a {single} sample $Y=R_{\ell}X + \sigma Z$, where $X\sim \m{N}(0,I)$, {$R_\ell\sim\Unif(\{0,\ldots,L-1\})$} is a random uniform shift, $Z\sim \m{N}(0,I)$, and $R_{\ell}$, $Z$ and $X$ are mutually independent.  The goal, then, is to recover $R_\ell$ from $X$ and $Y$.\footnote{A more general setting, where $X$ is not necessarily Gaussian, and $R_{\ell}X$ goes through some general channel, not necessarily Gaussian,  was studied by Wang, Hu, and Shayevitz~\cite{ws17}, but under different asymptotics.}

While the template matching problem seems to be significantly easier than the MRA problem, we show a surprising phenomenon: in high dimensions, template matching and MRA share the exact same phase transition point. 
In particular, it turns out that in high dimensions, under our parameterization $\sigma^2(\alpha)$, which amounts to $L/\sigma^2 = \alpha\log(L)$, the template matching problem displays a \emph{sharp recoverability threshold}. That is: 
(i) whenever $\alpha>2$, the random shift can be recovered with error probability $p_{e}\to 0$ as $L\to\infty$; (ii) whenever $\alpha<2$, the shift cannot be consistently recovered, and in fact for any estimator, $p_{e}\to 1$. 

Observe that the optimal estimator (in the sense of maximum a posteriori probability) for $R_\ell$ is given by:
\begin{equation}
	\widehat{R}_{\mathrm{MAP}} = \argmin_{\ell'}\|X-R_{\ell'}^{-1}Y\|^2 = \argmax_{\ell'} \frac{\langle X, R_{\ell'}^{-1}Y \rangle}{\|X\|^2} \,.
\end{equation}
Denote its error probability by 
\begin{equation}
	p_{e} = \Pr\left( R_{\ell} \ne \widehat{R}_{\mathrm{MAP}} \right) .
\end{equation}

We start by establishing that with overwhelming probability, the template $X$ is ``incoherent'', in the sense that the correlations $\langle X,R_{\ell'}X\rangle/\|X\|^2$ are very small, unless $\ell'=0$. The lemma is proved in Appendix~\ref{sec:proof_lemma_signal-is-nice}.

\begin{lemma}\label{lem:signal-is-nice}
	For $\kappa>0$, let $\m{A}(\kappa)$ be the event that 
	\[
	\left| L^{-1}\|X\|^2-1 \right| < \kappa \quad \textrm{ and }\quad \max_{\ell'\ne 0}L^{-1}\left| \langle X, R_{\ell'}X \rangle \right| \le \kappa
	,
	\]
	and let $\overline{\m{A}(\kappa)}$ be its complement.
	Then, 
	\[
	\Pr(\overline{\m{A}(\kappa)}) \le 2L\exp\left( -c L\min(\kappa,\kappa^2) \right),
	\] 
	for a universal constant $c>0$. 
	In particular, one can choose a sequence $\kappa=\kappa_L$ such that $\kappa\to 0$ sufficiently slowly, for example, $\kappa=CL^{-1/2}\log(L)$ for $C>0$ large enough, so that {$\Pr(\m{A}_{L}(\kappa_{L})) =1- o(1)$.} 
\end{lemma}

Let 
\begin{equation}
	\Theta_{\ell'} 
	= \frac{\langle X,R_{\ell'}^{-1}Y\rangle}{\|X\|^2} 
	= \frac{\langle X,R_{\ell-\ell'}X\rangle}{\|X\|^2} + \frac{\sigma \langle X,R_{\ell'}^{-1}Z \rangle}{\|X\|^2} ,
\end{equation}
and
\begin{equation}
	W_{\ell'} = \|X\|^{-1}\langle X,R_{\ell'}^{-1}Z \rangle .
\end{equation}
Recalling that
$
\widehat{R}_{\mathrm{MAP}} = \argmax_{\ell'} \Theta_{\ell'}
$,
and plugging $\sigma^2=(\alpha \log(L))^{-1}L$, Lemma~\ref{lem:signal-is-nice} implies that with high probability, 
\begin{equation}\label{eq:MLE-thetas}
	\Theta_{\ell'} = \begin{cases}
		1 + (1+ o(1)) \frac{1}{\sqrt{\alpha \log(L)}}\cdot W_\ell\quad&\textrm{ if }\ell'=\ell, \\
		o(1) + (1+ o(1)) \frac{1}{\sqrt{\alpha\log(L)}} \cdot W_{\ell'} \quad&\textrm{ if }\ell'\ne \ell. 
	\end{cases}
\end{equation}

\revDel{Since for every $\ell'$, $W_{\ell'}\sim \m{N}(0,1)$, it is obvious that $\Theta_{\ell}\limp 1$ as $L\to\infty$.}
\revAdd{Notice that for every $\ell'$, $W_{\ell'}\sim \m{N}(0,1)$, being the projection of $R_{\ell'}^{-1}Z\sim \m{N}(0,I)$ onto a unit vector $X/\|X\|$. This clearly implies that $\Theta_{\ell}\limp 1$ as $L\to\infty$.}
Thus, to analyze the error of the MAP estimator, it simply remains to understand the behavior of $\max_{\ell'}W_{\ell'}$.  
To this end, we recall the following three results. 
We start with a well-known fact about the maximum of i.i.d.\ standard Gaussians:
\begin{lemma}\label{lem:gaussain-max-exp}
	Let $Z_1,\ldots,Z_L$ be i.i.d $\m{N}(0,1)$ random variables. Then, as $L\to\infty$,
	\[
	\Expt\left[\max_{\ell}Z_\ell \right]/\sqrt{2\log(L)} \to 1 .
	\]
\end{lemma} 
The upper bound $\Expt\left[\max_{\ell}Z_l\right] \le \sqrt{2\log(L)}$ is elementary, and holds even when $Z_1,\ldots,Z_L$ are not independent. The proof follows from $\Expt\max_{\ell}Z_\ell \le \beta^{-1}\log\Expt \max_{\ell}e^{\beta Z_\ell} \le \beta^{-1} \log\Expt \sum_{\ell=1}^L e^{\beta Z_\ell} = \beta/2 + \beta^{-1}\log(L)$, which holds for all $\beta>0$; now take $\beta=\sqrt{2\log(L)}$. The proof of the matching  lower bound, on the other hand, is  more involved and follows from results in extreme value theory, see, for instance, Example 1.1.7 in~\cite{de2007extreme}. We also use the following ``quantitative'' version of the Sudakov-Fernique inequality: 
\begin{lemma}[Theorem 2.2.5 in \cite{adler2009random}]\label{lem:sudakov-fernique}
	Let $(X_1,\ldots,X_L)$ and $(Y_1,\ldots,Y_L)$ be Gaussian vectors so that ${\Expt[X_i]=\Expt[Y_i]}$ for all $i$. Set 
	\[
	\gamma_{i,j}^X = \Expt(X_i-X_j)^2,\quad \gamma_{i,j}^Y = \Expt(Y_i-Y_j)^2,
	\]
	and $\gamma=\max_{i,j}|\gamma_{i,j}^X-\gamma_{i,j}^Y|$. Then 
	\[
	\left| \Expt\left[\max_i X_i\right]-\Expt\left[\max_i Y_i\right]\right| \le \sqrt{2\gamma\log(L)} .
	\] 
\end{lemma}
To get concentration around the mean, we use (a simple case of) the Borell-TIS inequality:
\begin{lemma}
	\label{lem:borell-tis}
	Let $(X_1,\ldots,X_L)$ be a Gaussian vector, and set $\sigma^2 = \max_i \Expt[X_i^2]$. Then
	\[
	\Pr\left( \left|  \max_{i} X_i - \Expt\left[ \max_i X_i \right] \right| \ge t  \right) \le 2e^{-t^2/2\sigma^2} .
	\]
\end{lemma} 
See, e.g., \cite[Theorem 2.1.1]{adler2009random} (there only a one sided bound is stated; the other side follows the same way). 
The following is now an immediate corollary of Lemmas~\ref{lem:signal-is-nice}, \ref{lem:gaussain-max-exp},\ref{lem:sudakov-fernique} and \ref{lem:borell-tis}:
\begin{theorem}[Sharp threshold for template matching]\label{thm:template-sharp-thresh}
	If $\alpha>2$, then $p_{e}\to 0$ as $L\to\infty$. Conversely, if $\alpha<2$, then $p_{e}\to 1$. 
	\label{thm:template}
\end{theorem}
\begin{proof}
	We start by estimating $\Expt \max_{\ell'} W_{\ell'}$. Choose $\kappa=o(1)$ such that the event $\m{A}(\kappa)$ of Lemma~\ref{lem:signal-is-nice} holds with probability $1-o(1)$. 
	Conditioned on $X$, $\{ W_{\ell'} \}_{\ell'=0,\ldots,L-1}$ is a centered Gaussian vector, with covariance 
	\[
	C_{i,j}(X)=\Expt[W_i W_j\,\big|\,X] = \|X\|^{-2}\langle R_i X, R_j X\rangle,
	\]
	whereby under $\m{A}$, $\left| C_{i,j}(X)-\delta_{i,j} \right| =o(1)$.
	
	Let $(\tilde{W}_1,\ldots,\tilde{W}_{L-1})$ be i.i.d $\m{N}(0,1)$ random variables. By Lemmas~\ref{lem:gaussain-max-exp} and \ref{lem:sudakov-fernique}, conditioned on $X$ and under $\m{A}$,
	\[
	\Expt[\max_{\ell'}W_{\ell'}\,\big|\,X,\m{A}]= \Expt[\max_{\ell'}\tilde{W}_{\ell'}] + o(\sqrt{\log(L)}) = \sqrt{ (2{+}o(1))\log(L)} .
	\] 
	Lemma~\ref{lem:borell-tis} gives us a uniform (in $X$) concentration inequality, conditioned on $X$ and under $\m{A}$,
	\[
	\Pr\left( \left|\max_{\ell'}W_{\ell'} - \sqrt{2\log(L)}\right| \ge \sqrt{\varepsilon\log(L)} \,\Big|\, X,\m{A}\right) \le 2L^{-(\varepsilon{+}o(1))/2},
	\] 
	so that 
	\[
	\Pr\left( \left|\max_{\ell'}W_{\ell'} - \sqrt{2\log(L)}\right| \ge \sqrt{\varepsilon\log(L)} \right) \le 2L^{-(\varepsilon{+}o(1))/2} + \Pr\left(\overline{\m{A}}\right)  = o_{\varepsilon}(1).
	\]
	Thus, we have shown that $\max_{\ell'}W_{\ell'}/\sqrt{2\log(L)}\limp 1$. Using equation~\eqref{eq:MLE-thetas}, we deduce that $\Theta_{\ell}\limp 1$ whereas $\max_{\ell'\neq \ell}\Theta_{\ell'} \limp \sqrt{2/\alpha}$. 
	Since $\widehat{R}_{\mathrm{MAP}}=\argmax_{\ell'}\Theta_{\ell'}$, we conclude  that $p_e\to 0$ when $\alpha>2$ and $p_e\to 1$ when $\alpha<2$.  
\end{proof}

\paragraph{A remark on the relation between template matching and synchronization.}
In the MRA model, one does not have access to the true template and thus needs to estimate the relative shifts based solely on the data{; this problem is referred to as \emph{synchronization.} }



For simplicity, let us assume we are given two measurements $Y_1=X+\sigma Z_1$ and $Y_2=R_{\ell}X + \sigma Z_2$, and would like to estimate $R_{\ell}$ (recall that $X$ is unknown). The optimal (MAP) estimator is $\widehat{R}_{\mathrm{syn}}  
= \argmax_{\ell'} \Pr(R_{\ell'}|Y_1,Y_2) $.  
It is straightforward to show that  
\begin{align*}
	\widehat{R}_{\mathrm{syn}} 
	&= \argmax_{\ell'} \langle Y_1,R_{\ell'}^{-1}Y_2 \rangle 
	= \argmax_{\ell'} \langle (X+\sigma Z_1),R_{\ell'}^{-1}(R_{\ell}X+\sigma Z_2) \rangle	 \\
	&= \argmax_{\ell'} \left\{ \langle X,R_{\ell-\ell'}X\rangle + \sigma\langle X,R_{\ell'}^{-1}Z_2\rangle + \sigma\langle X,R_{\ell-\ell'}^{-1}Z_1\rangle + \sigma^2 \langle Z_1,R_{\ell'}^{-1}Z_2 \rangle \right\}. 
\end{align*}
In order for this to consistently return the true relative shift $R_{\ell}$, one needs to ensure that the ``noise'' term,
\[
\sigma\langle X,R_{\ell'}^{-1}Z_2\rangle + \sigma\langle X,R_{\ell-\ell'}^{-1}Z_1\rangle + \sigma^2 \langle Z_1,R_{\ell'}^{-1}Z_2\rangle
\]
is small compared to $\|X\|^2\sim L$. The ``typical'' size of the first two terms is
$
\sigma\langle X,R_{\ell'}^{-1}Z_2\rangle + \sigma\langle X,R_{\ell-\ell'}^{-1}Z_1\rangle \sim \sigma\sqrt{L} 
$,  
whereas the third is $\sigma^2 \langle Z_1,R_{\ell'}^{-1}Z_2\rangle \sim \sigma^2\sqrt{L}$, and is therefore the dominant one for large $\sigma$. Thus, to succeed with non-vanishing probability, we need that $\sigma^2\sqrt{L} \lessapprox L$, that is, $\sigma^2 \lessapprox \sqrt{L}$. In the regime we are interested in, the noise level is $\sigma^2\sim L/\log(L)$, and this turns out to be far too large. 

We mention in passing that if many measurements are available, one can  leverage the redundancy in the data to recover the true relative shifts in challenging environments; see for example~\cite{singer2011angular,singer2011three,boumal2016nonconvex,perry2018message,romanov2019noise}.

\section{Sample complexity lower bounds}
\label{sect:ITbounds}

\subsection{The information-theoretic method for estimation lower bounds}
We employ a standard information-theoretic method of obtaining estimation error lower bounds, via rate-distortion theory {(see e.g.~\cite{polyanskiy2014lecture})}. We refer the reader to SI Appendix~\ref{sec:ITback} for a basic review of the information-theoretic definitions and facts we use in this section. Let $\widehat{X}$ be an estimator of $X$ from the measurements $Y^n=(Y_1,\ldots,Y_n)$, which achieves expected error (``distortion'')
\begin{equation}\label{eq:RDF}
	\Expt \rho(X,\widehat{X}) = L^{-1}\Expt\min_{\ell=0,\ldots,L-1}\|X-R_\ell^{-1}\widehat{X}\|^2 \le \varepsilon .
\end{equation}
Since the estimator depends only on the measurements, and not on $X$,
the triplet $X-Y^n-\widehat{X}$ constitutes a Markov chain. Hence, by the data processing inequality (Proposition~\ref{prop:MIproperties} item \ref{mi:dpi}) we have that  $I(X;\widehat{X})\le I(X;Y^n)$. 
\revAdd{We lower-bound $I(X;\widehat{X})$ by the \emph{rate distortion function} (RDF) $R(\cdot)$ associated with the source $X\sim \m{N}(0,I)$, 
	and distortion measure $\rho(\cdot,\cdot)$: 
	\[
	R(\varepsilon) = \min_{P_{W|X} : \Expt\rho(X,W)\le \varepsilon} I(X;W) .
	\]
	The minimization here is done over conditional distributions $P_{W|X}$, or equivalently, over joint distributions $P_{X,W}$ whose $X$-marginal is $P_X$---in our case $\m{N}(0,I)$---obeying the average distortion constraint $\Expt\rho(X,W)\le \varepsilon$. Since the conditional distribution $P_{\widehat{X}|X}$ is, by definition, feasible for this minimization problem, we have $R(\varepsilon)\le I(X;\widehat{X})$.  
	Combining this with the upper bound $I(X;\hat{X})\le I(X;Y^n)$, we get
	\begin{equation}\label{eq:RDF-less-than-MI}
		R(\varepsilon) \le I(X;Y^n) ,
	\end{equation}
	and we shall next derive a lower bound for $R(\varepsilon)$ in terms of $\varepsilon$. 
}

\subsection{A lower bound on the rate-distortion function}

We start by obtaining a lower bound on the RDF. 
While the RDF problem for a Gaussian source under MSE distortion measure is classical, the MSE up to the best alignment (the distortion measure we consider) is somewhat non-standard.
Obtaining a precise expression for the true RDF seems difficult, but a simple lower bound can be obtained as follows. 

\begin{proposition}
	For an $L$ dimensional i.i.d. Gaussian vector $X\sim\m{N}(0,I)$, and distortion measure $\rho(\cdot,\cdot)$ as defined in~\eqref{eq:rho_def}, the rate distortion function satisfies
	\begin{align*}
		R(\varepsilon) \geq \frac{L}{2}\log\left(\frac{1}{\varepsilon}\right)-\log(L).
	\end{align*}
	\label{prop:mrardf}
\end{proposition}

\begin{proof}
	By definition of the rate distortion function, to establish the claim we need to show that for any conditional distribution (``test-channel'') $P_{W|X}$ that satisfies the constraint 
	$\Expt\rho(X,W)\le \varepsilon$, where $\rho(X,W)=L^{-1}\min_{\ell=0,\ldots L-1}\|X-R_{\ell}^{-1}W\|^2$, it holds that $I(X;W) \geq \frac{L}{2}\log\left(\frac{1}{\varepsilon}\right)-\log(L)$. To that end, let $R=R(X,W)=\argmin_{\ell'\in [0,\ldots,{L-1}]}\|X-R_{\ell'}W\|$ be the difference minimizing shift. 
	By the chain law of MI (Proposition~\ref{prop:MIproperties} item \ref{mi:chainrule}), 
	\begin{align}
		I(X;W)
		=I(X;W,R)-I(X;R|W)
		\ge I(X;W,R)-\log(L),\label{eq:mrardf1}
	\end{align}
	where we used $I(X;R|W)\le H(R|W)\le \log(L)$; the former follows from the definition of MI and non-negativity of entropy (Proposition~\ref{prop:entropy} item \ref{entropy:noneg}), and the latter follows from Proposition~\ref{prop:entropy} item \ref{entropy:unif} as the random variable $R$ can take at most $L$ values. Recall that $L^{-1}\mathbb{E}\|X-RW\|^2\leq \varepsilon$ by definition of $R$. We therefore have that
	\begin{align*}
		I(X;RW)\geq \min_{P_{W'|X}:L^{-1}\mathbb{E}\|X-W'\|^2\leq \varepsilon} I(X;W')=\frac{L}{2}\log\left(\frac{1}{\varepsilon}\right),
	\end{align*}
	where in the second equality we have used the well-known expression for the quadratic Gaussian rate distortion function (Proposition~\ref{prop:GaussianRDF}). Thus, using the data processing inequality (Proposition~\ref{prop:MIproperties} item \ref{mi:dpi}), we have
	\begin{align*}
		I(X;W,R)\geq I(X;RW)\geq \frac{L}{2}\log\left(\frac{1}{\varepsilon}\right).
	\end{align*}
	Substituting this into~\eqref{eq:mrardf1} establishes the claim.
\end{proof}

Combining Proposition~\ref{prop:mrardf} with equation~\eqref{eq:RDF-less-than-MI}, we get 
\revAdd{
	\[
	I(X;Y^n) \ge R(\varepsilon) \ge \frac{L}{2}\log\left(\frac{1}{\varepsilon}\right)-\log(L)\,.
	\]
	Setting $\varepsilon = \Expt\rho(X,\widehat{X})$, we have obtained the following bound: 
}
\begin{corollary}\label{prop:IT-lower-bound}
	Suppose that $X\sim \m{N}(0,I)$ is an $L$ dimensional i.i.d.\ Gaussian vector, $\widehat{X}$ is any estimator of $X$ from $Y_1,\ldots,Y_n$, and $\rho(\cdot,\cdot)$ is as defined in~\eqref{eq:rho_def}. Then
	\[
	\Expt \rho(X,\widehat{X}) \ge \exp \left(- \frac{2I(X,Y^n) + 2\log(L)}{L} \right) = \exp\left(-2L^{-1}\cdot I(X,Y^n) {+} o(1) \right).
	\] 
	Equivalently,
	\begin{align*}
		\MSE(L,\alpha,n)\ge \exp \left(- \frac{2I(X,Y^n) + 2\log(L)}{L} \right) = \exp\left(-2L^{-1}\cdot I(X,Y^n) {+} o(1) \right).
	\end{align*}
\end{corollary}

\revAdd{Corollary \ref{prop:IT-lower-bound} tells us that an upper bound on the {MI} $I(X;Y^n)$ would give us a lower bound on the expected error of any estimator  of $X$ from $Y^n=(Y_1,\ldots,Y_n)$. We devote the next section to deriving such upper bounds.}

\subsection{Upper bounds on the mutual information}

\revDel{In light of Corollary~\ref{prop:IT-lower-bound}, an upper bound on the {MI} $I(X;Y^n)$ provides a lower bound on the expected error of any estimator of $X$ from $Y^n=(Y_1,\ldots,Y_n)$.}

We start with the rather trivial observation that the MI between the signal~$X$ and the measurements~$Y^n$ is smaller than the MI in a problem where there are no random shifts, which is equal to  $\frac{L}{2}\log(1+n\sigma^{-2})$. 
The next lemma formalizes this intuition and quantifies the MI difference between the two problems.

\begin{lemma}\label{lem:gaussian-MI}
	The mutual information between the signal $X$ and measurements $Y_1,\ldots,Y_n$ is
	\begin{align}
		I(X;Y^n) = \frac{L}{2}\log(1+n\sigma^{-2}) - I(R^n;X|Y^n),\label{eq:alphageq1MIbound}
	\end{align}
	where $R^n=(R_{\ell_1},\ldots,R_{\ell_n})$. In particular, $I(X;Y^n) \le \frac{L}{2}\log(1+n\sigma^{-2})$.
\end{lemma}
\begin{proof}
	Let $\tilde{Y}_i = R_{\ell_i}^{-1}Y_i=X+\sigma R_{\ell_i}^{-1}Z_i$. We may write
	\begin{align*}
		I(X;Y^n)&=I(X;Y^n,R^n)-I(X;R^n|Y^n)\\
		&=I(X;\tilde{Y}^n,R^n)-I(X;R^n|Y^n)\\
		&=I(X;\tilde{Y}^n)+I(X;R^n|\tilde{Y}^n)-I(X;R^n|Y^n),
	\end{align*}
	where the first and third equalities follow by the chain rule for MI (Proposition~\ref{prop:MIproperties} item \ref{mi:chainrule}), and the second follows from Proposition~\ref{prop:MIproperties} item \ref{mi:invertible}, and the fact that the mapping $(Y^n,R^n)\mapsto (\tilde{Y}^n,R^n)$ is invertible. By the fact that the Gaussian distribution is rotation invariant, and in particular $R^{-1}_{\ell_i} Z\sim\m{N}(0,I)$, we have that $R^n$ is statistically independent of $(X,\tilde{Y}^n)$, and consequently
	\begin{align*}
		I(X;R^n|\tilde{Y}^n)=H(R^n|\tilde{Y}^n)-H(R^n|\tilde{Y}^n,X)=H(R^n)-H(R^n)=0,
	\end{align*}
	where the first equality follows by definition of conditional mutual information and the second by Proposition~\ref{prop:MIproperties}.\ref{entropy:concavity}. It remains to compute $I(X;\tilde{Y}^n)$. 
	\revDel{To this end, note that $P_{\tilde{Y}^n|X=x}=\m{N}^{\otimes n}(x,\sigma^2 I)$, that is, $X$, and $\tilde{Y}^n$ have the same joint distributed as $X$ and $(X+\sigma Z_1,\ldots,X+\sigma Z_n)$, i.e., as $n$ measurements of a signal in AWGN. It is well known that the sample average $\frac{1}{n}\sum_{i=1}^n X+\sigma Z_i$ is a sufficient statistic of $(X+\sigma Z_1,\ldots,X+\sigma Z_n)$ for $Y$. We therefore have that
		\begin{equation}
			\begin{split}
				I(X;\tilde{Y}^n) &= I(X;X+\sigma Z_1,\ldots,X+\sigma Z_n)= I\left(X;X+\frac{\sigma}{n}\sum_{i=1}^n  Z_i\right) \\ &= I\left(X; X + \m{N}(0,(\sigma^2/n) I)\right) = \frac{L}{2}\log(1+n\sigma^{-2}),
			\end{split}	
	\end{equation}}
	\revAdd{To this end, note that conditioned on $X=x$, the measurements $\tilde{Y}_1,\ldots,\tilde{Y}_n$ are simply i.i.d. Gaussian measurements $Y_i\sim \m{N}(x,\sigma^2 I)$. It is well-known that in this case, the sample mean $\frac{1}{n} \sum_{i=1}^n \tilde{Y}_i = X$ is a sufficient statistic of $\tilde{Y}^n$ for $X$. Conditioned on $X=x$, the sample mean has distribution $\frac{1}{n}\sum_{i=1}^n \tilde{Y}_i \sim \m{N}(x,\sigma^2/n \cdot I)$, therefore,
		\begin{equation}
			\begin{split}
				I(X;\tilde{Y}^n) = I\left(X;\frac1n \sum_{i=1}^n \tilde{Y}_i \right)  = I\left(X; X + \m{N}(0,\sigma^2/n \cdot I)\right) = \frac{L}{2}\log(1+n\sigma^{-2}),
			\end{split}	
	\end{equation}}
	
	where the last equality follows from Proposition~\ref{prop:MIproperties} item \ref{mi:Gaussian}.
\end{proof}

Combining Corollary~\ref{prop:IT-lower-bound} and Lemma~\ref{lem:gaussian-MI}, we obtain the following lower bound, that essentially says the MSE in the MRA model is no better than in estimating a signal in AWGN.

\begin{corollary}
	\label{cor:GaussianMSE_LB}
	The smallest attainable MSE in the MRA model satisfies
	\begin{align*}
		\MSE(L,\sigma^2,n)\geq \frac{L^{-\frac{2}{L}}}{1+n\sigma^{-2}}=\frac{1}{1+n\sigma^{-2}}(1+o(1)),
	\end{align*}
	and the sample complexity satisfies
	\begin{align*}
		\SC(L,\sigma^2,\varepsilon)\geq \left\lceil \left(\frac{L^{-\frac{2}{L}}}{\varepsilon}-1\right)\sigma^2\right\rceil=n^*_{\text{AWGN}}(L,\sigma^2,\varepsilon)(1+o(1)).
	\end{align*}
\end{corollary}

Lemma~\ref{lem:gaussian-MI} tells us that the gap between $I(X;Y^n)$ and the MI in estimating a signal in AWGN, without the shifts, $\frac{L}{2}\log(1+n\sigma^{-2})$, is $I(X;R^n|Y^n)$. This quantity is intimately related to a multi-sample version of the template matching problem, as was considered in Section~\ref{sect:template-matching}. 
This connection will be exploited later on, when we derive an upper bound on the {single sample} MI $I(X;Y_i)$.


\paragraph{Information combining} Observe that the measurements $Y_1,\ldots,Y_n$ are mutually independent \revAdd{and identically distributed} conditioned on $X$; that is, the samples are obtained by passing the same signal $X$ independently through a memoryless channel\revDel{ $P_{Y^n|X}=P^{\otimes n}_{Y|X}$}. By Proposition~\ref{prop:MIproperties} item \ref{mi:memoryless}, this implies that
\begin{align}
	I(X;Y^n)\leq \sum_{i=1}^n I(X;Y_i)=nI(X;Y),\label{eq:MImemoryless}
\end{align} 
where $Y=R_\ell X+\sigma Z$ is a single measurement in the MRA model. Substituting~\eqref{eq:MImemoryless} into Corollary~\ref{prop:IT-lower-bound}, yields the following.

\begin{proposition}
	The smallest attainable MSE in the MRA model satisfies
	\begin{align*}
		\MSE(L,\sigma^2,n)\geq L^{-\frac{2}{L}}\exp\left(-n\frac{2}{L}I(X;Y)\right)=\exp\left(-n\frac{2}{L}I(X;Y)\right)(1+o(1)),
	\end{align*}
	and the sample complexity satisfies
	\begin{align*}
		\SC(L,\sigma^2,\varepsilon)\geq \frac{L}{2}\cdot\frac{\log\left(\frac{1}{\varepsilon}\right)-\frac{2\log{(L)}}{L}}{I(X;Y)}=\log\left(\frac{1}{\varepsilon}\right)\cdot \frac{L}{2I(X;Y)}(1+o(1)),
	\end{align*}
	where $Y=R_\ell X+\sigma Z$ is a single measurement in the MRA model.
	\label{prop:singlesampleMIbasedBounds}
\end{proposition} 


It is important to emphasize at this point that the bound in~\eqref{eq:MImemoryless} becomes very loose for $n$ sufficiently large. Indeed, Lemma~\ref{lem:gaussian-MI} implies that $I(X;Y^n)$ should scale at best logarthmically, rather than linearly, with~$n$. Consequently, the lower bound on $\MSE(L,\sigma^2,n)$ in Proposition~\ref{prop:singlesampleMIbasedBounds} decreases exponentially fast with $n$, whereas we know from Corollary~\ref{cor:GaussianMSE_LB} that it cannot decrease faster than the parametric rate of $1/n$ as in estimating a signal in AWGN.
Despite its grossly wrong dependence on $n$, the upper bound $I(X;Y^n)\leq nI(X;Y)$ {does} suffice to say something non-trivial about the sample complexity of the problem. As seen from Proposition~\ref{prop:singlesampleMIbasedBounds}: in order for the estimation error to be {strictly bounded away from one}, one needs at least $\Omega(L\cdot I(X;Y)^{-1})$ samples. We will see that this rather ``na\"ive'' analysis is already enough to accurately separate between a ``high SNR'' and a ``low SNR'' regime, where the behavior of the MRA problem is qualitatively different. 
Intuitively, as the measurements $Y_1,\ldots,Y_n$ are only dependent through the random variable $X$, if $n$ is so small that it is impossible to learn much about $X$ from $Y^n$, the dependence between $Y_1,\ldots,Y_n$ must be weak. Thus, in that regime, ignoring this dependence and bounding $I(X;Y^n)\leq n I(X;Y)$ is a rather accurate estimate. 

The problem of obtaining a stronger bound on multi-sample MI $I(X;Y^n)$ in terms of the single-sample MI $I(X;Y)$ is an instance of a so-called \emph{information combining} problem. Several problems of this type have been studied in the information theory literature, mostly dealing with binary channels~\cite{ssz05,lhhh05}. In our case, we believe this problem to be quite hard, at least in the low SNR regime, and thus  we could not obtain a tighter bound. Deriving such bounds can yield stronger lower bounds on $\MSE(L,\alpha,n)$ in the low-SNR regime ($\alpha<2$) than the ones we obtain here using the simple bound $I(X;Y^n)\leq n I(X;Y))$.

\paragraph{Roadmap} We will devote the rest of this section to deriving upper bounds on $I(X;Y)$. These bounds, together with Proposition~\ref{prop:singlesampleMIbasedBounds}, will immediately imply lower bounds on  the MSE and the sample complexity. 
In particular, we will derive two bounds, using different methods, that will be effective in two SNR regimes. 
\begin{itemize}
	\item We estimate the mutual information using Jensen's inequality to facilitate the computation of several expectations. One could expect this method to give somewhat tight results when~$I(X;Y)$ is very small, and indeed, we shall see that when $0<\alpha < 1$, we obtain a bound $I(X;Y)=O(L^{\alpha-1})$, which tends to $0$ as $L\to\infty$. For $\alpha\ge 1$, the obtained bound will turn out to be too loose. 
	
	\item In Lemma~\ref{lem:gaussian-MI} we have found that $I(X;X+\sigma Z)-I(X;Y)=I(X,R_{\ell}|Y)$. We lower bound this gap using a Fano-like inequality, which in the case $\alpha<2$ amounts to ``quantifying'' how well $R_{\ell}$ can be estimated from $X$ and $Y$, in a somewhat more precise sense than Theorem~\ref{thm:template-sharp-thresh} (which tells us that in this case, the error is $p_{e}=1-o(1)$). This will allow us to show that when $\alpha<2$, $I(X;Y)=o(\log(L))$. We will not, however, be able to recover the estimate  in the case of $0<\alpha<1$ using this method.
	
	
\end{itemize}

\subsubsection{MI bound at very low SNR ($\alpha<1$)}

We first express $I(X;Y)$ in the following way: 
\begin{lemma}\label{lem:MI-calculation-step-1}
	Suppose that $X\sim\m{N}(0,I)$, $Z\sim \m{N}(0,I)$, and {$R\sim\Unif(\{R_0,\ldots,R_{L-1}\})$} are mutually independent. Then, 
	\[
	I(X;Y) = \frac{L}{2}\log(1+\sigma^{-2})-L\sigma^{-2} + \Expt_{X,Z} \left[ \log \Expt_{R} \exp\left( \frac{1}{\sigma^2} \langle X+\sigma Z,RX \rangle \right) \right].
	\]
\end{lemma}
\begin{proof}
	Write $I(X;Y)=h(Y)-h(Y|X)$. Note that for any shift $R_\ell$, $R_{\ell}X\sim \m{N}(0,I)$ and therefore $Y\sim\m{N}(0,(1+\sigma^2)I)$; this means that $Y=R_\ell X+\sigma Z$ is independent of $R_\ell$. The differential entropy of $Y$ is 
	$h(Y)=h(\m{N}(0,(1+\sigma^2)I)=\frac{L}{2}\log(2\pi e) + \frac{L}{2}\log(1+\sigma^2)$, by Proposition~\ref{prop:entropy} item \ref{entropy:gauss}. 
	
	Let us now write the conditional differential entropy explicitly. The conditional density of $Y$ given~$X$ is $p_{Y|X}(y|x)=\Expt_{R}\left[ (2\pi\sigma^2)^{-L/2}\exp\left(-\frac{1}{2\sigma^2}\|y-R x\|^2\right) \right]$ for uniform~$R$. The conditional entropy is then simply
	\begin{align*}
		h(Y|X) 
		&= \Expt_{X,Y}\left[ -\log p_{Y|X}(Y|X)\right] \\
		&= \frac{L}{2}\log(2\pi \sigma^2) - \Expt_{X,Y} \left[ \log \Expt_{R} \exp\left( -\frac{1}{2\sigma^2}\|Y-RX\|^2 \right) \right] \\
		&= \frac{L}{2}\log(2\pi \sigma^2) - \Expt_{X,Y} \left[ \log \Expt_{R} \exp\left( -\frac{1}{2\sigma^2}\left( \|Y\|^2+\|X\|^2 - 2\langle Y,RX \rangle \right) \right) \right] \\
		&= \frac{L}{2}\log(2\pi \sigma^2) + \frac{L+(1+\sigma^2)L}{2\sigma^2}- \Expt_{X,Y} \left[ \log \Expt_{R} \exp\left( \frac{1}{\sigma^2} \langle Y,RX \rangle \right) \right] .
	\end{align*}
	\revAdd{It remains to compute the expectation with respect to the joint distribution of $X$ and $Y$ in the last term. Recall that we can write $Y=R'X + \sigma Z$ for $R' \sim \Unif(\{R_0,\ldots,R_{L-1}\}) $ and $Z\sim \m{N}(0,I)$, both independent of $X$. Alternatively, we could also write $Y=R'(X+\sigma Z)$, which defines the exact same joint distribution between $X$ and $Y$, due to the orthogonal invariance of $Z\sim \m{N}(0,I)$; this second form is slightly more convenient in what follows.}\revDel{ We can write $Y=R'(X+\sigma Z)$, where $R'$ is another uniform shift (independent of $X,Y,R$); here we used the orthogonal invariance of $Z\sim\m{N}(0,I)$.} Since $R$ is uniformly distributed, 
	\begin{align*}
		\Expt_{X,Z,R'} \left[ \log \Expt_{R} \exp\left( \frac{1}{\sigma^2} \langle R'(X+\sigma Z),RX\rangle  \right) \right]
		&= \Expt_{X,Z,R'} \left[ \log \Expt_{R} \exp\left( \frac{1}{\sigma^2} \langle (X+\sigma Z),(R')^{-1}RX \rangle \right) \right] \\
		&= \Expt_{X,Z} \left[ \log \Expt_{R} \exp\left( \frac{1}{\sigma^2} \langle (X+\sigma Z),RX \rangle \right) \right],
	\end{align*}
	that is, we can ``drop'' $R'$. The claimed formula now readily follows. 
\end{proof}

The following proposition is the main estimate of this section. The proof uses some properties of the spectrum of $R_\ell$, stated and proved in~Appendix~\ref{sec:spectrum_shift_operator}. 

\begin{proposition}\label{prop:MI-at-low-SNR}
	We have the following upper bound on the single sample MI:
	\[
	I(X;Y) \le \log\left(1+L^{-1}e^{\sigma^{-2}L}\right) + O(\sigma^{-4}L) .
	\]
	In particular, if $\sigma^{-2}L=\alpha\log(L)$ for $0<\alpha<1$, then the MI asymptotically vanishes as $L\to\infty$ with $I(X;Y)\leq L^{-1+\alpha}(1+o(1))$. 
\end{proposition}
\begin{proof}
	By the concavity of the $\log$ function, we always have $\Expt_W \log(W) \le \log(\Expt W)$. Thus,
	\begin{align*}
		\Expt_{X,Z} \left[ \log \Expt_{R} \exp\left( \frac{1}{\sigma^2} \langle X+\sigma Z,RX \rangle \right) \right]
		&\le \Expt_{X} \left[ \log \Expt_{Z, R} \exp\left( \frac{1}{\sigma^2} \langle X+\sigma Z,RX \rangle \right) \right] \\
		&= \Expt_{X} \left[ \log \Expt_{R} \exp\left( \frac{1}{\sigma^2} \langle X,RX \rangle + \frac{1}{2\sigma^2}\|RX\|^2 \right) \right] \\
		&= \Expt_{X} \left[ \log \Expt_{R} \exp\left( \frac{1}{\sigma^2} \langle X,RX \rangle + \frac{1}{2\sigma^2}\|X\|^2 \right) \right] \\
		&= \frac{1}{2}\sigma^{-2}L + \Expt_{X} \left[ \log \Expt_{R} \exp\left( \frac{1}{\sigma^2} \langle X,RX \rangle \right) \right] \\
		&\le \frac{1}{2}\sigma^{-2}L + \log \Expt_{R, X} \exp\left( \frac{1}{\sigma^2} \langle X,RX \rangle \right) .
	\end{align*}
	Plugging into the expression in Lemma~\ref{lem:MI-calculation-step-1}, we get
	\[
	I(X;Y) \le \frac{L}{2}\log(1+\sigma^{-2})-\frac{1}{2}L\sigma^{-2} + \log \Expt_{R, X} \exp\left( \frac{1}{\sigma^2} \langle X,RX \rangle \right) .
	\]
	Note that as $L,\sigma^2\to \infty$, already $\frac{L}{2}\log(1+\sigma^{-2})-\frac{1}{2}L\sigma^{-2}=O(\sigma^{-4}L)$. Observe that $\langle X,RX \rangle = \langle X,R^\T X\rangle = \frac12 \langle X, (R+R^\T)X\rangle$. By Lemma~\ref{lem:DFT-basis}, all the matrices $R_\ell +R_\ell^\T$ are diagonalized
	by some orthonormal basis with eigenvalues $\{2\cos\left(\frac{2\pi}{L} k\ell \right)\}_{k=0}^{L-1}$. By the orthogonal invariance of $X\sim \m{N}(0,I)$, there are i.i.d.\ $W_{k,\ell} \sim \m{N}(0,1)$ such that for all $\ell$,
	\[
	\sigma^{-2} \langle X,R_\ell X\rangle = \sigma^{-2}\sum_{k=0}^{L-1}\cos\left(\frac{2\pi}{L} k\ell \right) W_{k,\ell}^2 .
	\]
	Recall that the moment generating function of a $\chi^2$ random variable is 	\[
	\Expt_{W\sim \m{N}(0,1)}[e^{tW^2}] = (1-2t)^{-1/2} \quad\textrm{ for } t>1/2\,,
	\]
	\revAdd{see, e.g, \cite[page 621]{casella2002statistical}.}
	Therefore, assuming $\sigma^2$ is sufficiently large (e.g., $\sigma^2>2$),
	\begin{align*}
		\log \Expt_{R, X} \exp\left( \frac{1}{\sigma^2} \langle X,RX \rangle \right)
		&= \log \left[ L^{-1}\sum_{\ell=0}^{L-1} \prod_{k=0}^{L-1} \left( 1-2\sigma^{-2}\cos\left(\frac{2\pi}{L} k\ell \right) \right)^{-1/2} \right] \\
		&= \log\sum_{\ell=0}^{L-1} e^{\psi_\ell} - \log(L),
	\end{align*}
	where 
	\[
	\psi_{\ell} = -\frac12 \sum_{k=0}^{L-1} \log\left( 1-2\sigma^{-2}\cos\left(\frac{2\pi}{L} k\ell \right)\right) .
	\]
	Expanding the $\log$ function to first order around $1$ and noting that $\sum_{k=0}^{L-1}\cos\left(\frac{2\pi}{L}k\ell\right)=L\cdot \Ind_{\{\ell=0\}}$ (see Lemma~\ref{lem:DFT-basis}), for large values of $L$ and $\sigma^2$, we get 
	\[
	\psi_\ell = \sum_{k=0}^{L-1} \sigma^{-2}\cos\left(\frac{2\pi}{L} k\ell \right) + O(\sigma^{-4}L) = 
	\begin{cases}
		\sigma^{-2}L + O(\sigma^{-4}L) \quad&\textrm{ if }\ell = 0, \\
		O(\sigma^{-4}L) \quad&\textrm{ otherwise.}
	\end{cases}
	\]
	Thus, we have the estimate
	\begin{align*}
		\log\sum_{\ell=0}^{L-1} e^{\psi_\ell}-\log(L) &= \log \left( \frac{1}{L}e^{\sigma^{-2}L + O(\sigma^{-4}L)} + \frac{L-1}{L}e^{O(\sigma^{-4}L)} \right) \\&= \log\left( 1+L^{-1}e^{\sigma^{-2}L} \right) + O(\sigma^{-4}L) ,
	\end{align*}
	from which the claimed result immediately follows.
\end{proof}





Observe that for $\alpha>1$, Proposition~\ref{prop:MI-at-low-SNR} gives an upper bound of the order $I(X;Y)=O(\log(L))$. It will turn out that when $\alpha>2$, this is indeed the right order of magnitude. However, for $1<\alpha\le 2$ the bound is too loose, and in fact $I(X;Y)=o(\log(L))$. 


\subsubsection{MI bound using template matching}
\label{subsec:MIviaTM}

We start from Lemma~\ref{lem:gaussian-MI} which gives, for $n=1$ and $Y=RX+\sigma Z$, $I(X;Y) = \frac{L}{2}\log(1+\sigma^{-2}) - I(R;X|Y)$. We make the important observation that $R$ and $Y$ are independent; indeed, regardless of $R$, it holds that $Y|R\sim \m{N}(0,(1+\sigma^2)I)$. 
We remark, however, that when $n>1$, $Y^n$ is {not} independent of $R^n$.
We can therefore use Proposition~\ref{prop:entropy} item \ref{entropy:concavity}, and Proposition~\ref{prop:entropy} item \ref{entropy:unif} to write
\[
I(R;X|Y)=H(R|Y)-H(R|X,Y)=H(R)-H(R|X,Y)=\log(L)-H(R|X,Y),
\]
so that 
\begin{equation}\label{eq:MI-entropy-gap}
	I(X;Y)=\frac{L}{2}\log(1+\sigma^{-2})-\log(L)+H(R|X,Y). 
\end{equation}

The following is now an immediate consequence of Fano's inequality (Proposition~\ref{prop:fano}) and Theorem~\ref{thm:template-sharp-thresh}. 
\begin{proposition}\label{prop:fano-immediate}
	Suppose that $\sigma^{-2}L = \alpha\log(L)$ with $\alpha>2$. Then,
	\begin{align*}
		I(X;Y) &= \frac{L}{2}\log(1+\sigma^{-2})-(1{+}o(1))\log(L) \\&= \left(\frac{\alpha}{2}-1+o(1)\right)\log(L) + O(\sigma^{-4}L) .
	\end{align*}
\end{proposition}
\begin{proof}
	We estimate $H(R|X,Y)$. Clearly, $H(R|X,Y)\ge 0$ by non-negativity of entropy (Proposition~\ref{prop:entropy} item \ref{entropy:noneg}). As for an upper bound,
	by Fano's inquality (Proposition~\ref{prop:fano}), for any estimator $\widehat{R}$ of $R$ from $X,Y$, the error probability $p_e=\Pr(R\ne \widehat{R})$ satisfies
	\[
	H(R|X,Y) \le \log{2}+ p_e \log(L).
	\]
	By Theorem~\ref{thm:template-sharp-thresh}, $\widehat{R}_{\mathrm{MAP}}$ has error $p_e\to 0$, which means that $H(R|X,Y)=o(1)\cdot \log(L)=o(\log(L))$. 
	Plugging this into equation~\eqref{eq:MI-entropy-gap} and expanding $\frac{L}{2}\log(1+\sigma^{-2})=\frac{\alpha}{2}\log(L) + O(\sigma^{-4}L)$, we obtain the desired estimate for $I(X;Y)$.   
\end{proof}

\revAdd{
	Proposition~\ref{prop:fano-immediate} above will not be needed for our main results, but its proof serves as good exposition towards bounding the conditional entropy $H(R|X,Y)$ in the harder case $\alpha\le 2$.
}
When $\alpha<2$ we have $p_e\to 1$, so that it is no longer true that $H(R|X,Y)=o(\log(L))$. Indeed, since $I(X;Y)=(\alpha/2-1)\log(L) + O(\sigma^{-4}L) + H(R|X,Y)$, we must have that $H(R|X,Y) \ge (1-\alpha/2-o(1))\log(L)$, since the MI is non-negative. While, indeed, in this regime $R$ cannot be recovered from $X,Y$, we can still obtain a non-trivial upper bound (of the form $c(\alpha)\log(L)$ for some $c(\alpha)<1$) on the conditional entropy $H(R|X,Y)$; the idea is that given $X,Y$, we can form a relatively small list that contains $R$ with high probability.

Our goal, then, is to non-trivially upper bound $H(R|X,Y)$ in the regime $\alpha\leq 2$ where $p_e \not\to 0$. Let $\tau>0$, and denote by $S_{\tau}$ the set of $\tau$-likely shifts:
\begin{equation}\label{eq:S-tau}
	\m{S}_{\tau} = \left\{ R'\,:\, \frac{\langle X,(R')^{-1}Y\rangle }{\|X\|^2} \ge 1-\tau \right\} .
\end{equation}
The analysis of Section~\ref{sect:template-matching} tells us that for any $\tau>0$, the true shift $R$ belongs with high probability to the set $\m{S}_{\tau}$. Moreover, when $\alpha>2$ (and $\tau>0$ is a sufficiently small constant), in fact with high probability $S_{\tau}=\{R\}$. When $\alpha\le 2$ this will no longer be the case; nonetheless, we show that $|\m{S}_\tau|$ is with high probability {significantly} smaller than $L$. This means that given~$X$ and~$Y$, we can produce a list of likely candidates for $R$ which is much smaller than the entire group of shifts. The following lemma is proved in the SI Appendix, Section~\ref{sec:proof_lem_not_many_candidates}.

\begin{lemma}\label{lem:not-many-candidiates}
	Let $\kappa,\tau,\zeta>0$. Set $M=L^{1-\frac12 \alpha(1-\kappa)\left(1-\tau-\frac{\kappa}{1-\kappa}\right)^2 + \zeta}$, and assume that $\alpha\le 2$. Then 
	\begin{align}
		\Pr\left( R\notin \m{S}_\tau \textrm{ or } |\m{S}_\tau|> M \right) \le 2Le^{-c L\min(\kappa,\kappa^2) }  + L^{-\frac12 \alpha(1-\kappa)\left(1-\tau-\frac{\kappa}{1-\kappa}\right)^2} + 2L^{-\zeta},\label{eq:listsizeUB}
	\end{align}
	where $c>0$ is the universal constant of Lemma~\ref{lem:signal-is-nice}. 
\end{lemma}

Lemma~\ref{lem:not-many-candidiates} implies that there are slowly decaying sequences $\tau=\tau_L=o(1),\delta=\delta_L=o(1)$ such that the event 
\[
\m{B} = \left\{ R\in \m{S}_{\tau_L}\textrm{ and } |\m{S}_{\tau_L}|\le L^{1-\frac12 \alpha + \delta_L} \right\}
\]
holds with high probability {of} $\Pr(\m{B})=1-o(1)$. We use this to bound the conditional entropy $H(R|X,Y)$, and obtain a bound on the MI:

\begin{proposition}\label{prop:MI-alpha-le-2}
	Suppose that $\alpha\le 2$. Then, 
	\[
	I(X;Y) = o(\log(L)) .
	\]
\end{proposition}
\begin{proof}
	We upper bound the conditional entropy $H(R|X,Y)$ using a ``Fano-like'' argument. Let $E$ be the indicator for the event $\m{B}$ above. Since $E$ is completely deterministic given $(R,X,Y)$, we have that $H(E|R,X,Y)=0$ by Proposition~\ref{prop:entropy} item \ref{entropy:noneg} and by the chain rule of entropy (Proposition~\ref{prop:entropy} item \ref{entropy:chainrule}) we have 
	\begin{align*}
		H(R|X,Y) &= H(R|X,Y)+H(E|R,X,Y)\\
		&= H(R,E|X,Y) \\
		&= H(E|X,Y) + H(R|X,Y,E) \\
		&\le H(E) + H(R|X,Y,E=1)\Pr(E=1) + H(R|X,Y,E=0)\Pr(E=0),
	\end{align*}
	where we have bounded $H(E|X,Y)\leq H(E)$ using Proposition~\ref{prop:entropy} item \ref{entropy:concavity}, and expanded $H(R|X,Y,E)$ according to the definition of conditional entropy, averaging only with respect to $E$.
	
	Now, given that $E=1$, we know that $R$ belongs to $\m{S}_{\tau_L}$, which has size $|\m{S}_{\tau_L}|\le M=L^{1-\frac12 \alpha + \delta_L}$. Hence, $ H(R|X,Y,E=1) \le \log(M)=\left(1-\frac12 \alpha + \delta_L \right)\log(L)$ by Proposition~\ref{prop:entropy} item \ref{entropy:unif}, and by the same reason $H(R|X,Y,E=0) \le \log(L)$. By definition, $\Pr(E=1)=\Pr(\m{B})=1-o(1)$, and $H(E)\le \log(2)$ by Proposition~\ref{prop:entropy} item \ref{entropy:unif}. Thus, $H(R|X,Y) \le \left( 1-\frac12 \alpha + o(1) \right)\log(L)$. Plugging this into Eq. \eqref{eq:MI-entropy-gap}, 
	\begin{align*}
		I(X;Y) 
		&= \frac{L}{2}\log(1+\sigma^{-2}) - \log(L) + H(R|X,Y) \\
		&= \left( \frac{\alpha}{2}-1+o(1) \right)\log(L) + O(\sigma^{-4}L) + \left(1-\frac{\alpha}{2} + o(1) \right)\log(L) \\
		&= o(\log(L)) + O(\sigma^{-4}L), 
	\end{align*} 
	as claimed.
\end{proof}

\begin{remark}
	One might wonder if the argument above (if carried out delicately enough) can match the estimate $I(X;Y)=O(L^{-1+\alpha})$ we have already seen for $\alpha<1$. Unfortunately, the bound $\Pr(|S_{\tau}|\ge M)\le 2L^{-\delta}$ (using Markov's inequality; see the proof of Lemma~\ref{lem:not-many-candidiates} in SI Appendix, Section~\ref{sec:proof_lem_not_many_candidates}) is already too crude for that purpose: since we need to choose $\delta=o(1)$, the $o(1)$ correction above must decay slower than $L^{-c}$ (for any $c>0$). 
\end{remark}

\subsubsection{Proof of main results}\label{sect:main-proofs}

We are ready to prove Theorem~\ref{thm:alphalessthan1} and the sample complexity lower bounds of Theorem~\ref{thm:highsnrregime}. 

\paragraph{Proof of Theorems~\ref{thm:highsnrregime} (lower bounds) and \ref{thm:alphalessthan1}.} 
\begin{itemize}
	\item Theorem~\ref{thm:highsnrregime}, $\alpha>2$ (lower bound): Corollary~\ref{cor:GaussianMSE_LB} immediately implies that 
	\[
	\lim_{\varepsilon\to 0}\lim_{L\to\infty}\frac{\SC(L,\alpha,\varepsilon)}{\sigma^2/\varepsilon}\ge 1\,.
	\]
	\item Theorem~\ref{thm:highsnrregime}, $\alpha\le 2$: Combining Proposition~\ref{prop:singlesampleMIbasedBounds} and Proposition~\ref{prop:MI-alpha-le-2},  give  
	\[
	\SC(L,\alpha,\varepsilon)=\omega\left(\frac{L}{\log(L)}\log(1/\varepsilon)\right) \revAdd{ = \omega\left(\sigma^2 \log(1/\varepsilon)\right)}\,.
	\]
	\item Theorem~\ref{thm:alphalessthan1}, $\alpha<1$: Combining Proposition~\ref{prop:singlesampleMIbasedBounds} and Proposition~\ref{prop:MI-at-low-SNR} yield 
	\[
	\SC(L,\alpha,\varepsilon)=\Omega(L^{2-\alpha}\log(1/\varepsilon))\,.
	\]
\end{itemize}
The proof of the upper bound $\lim_{\varepsilon\to 0}\lim_{L\to\infty}\frac{\SC(L,\alpha,\varepsilon)}{\sigma^2/\varepsilon}\le1$ for $\alpha>2$ (item (1) of Theorem~\ref{thm:highsnrregime}) appears in Section~\ref{sec:upper-bound}.

\section{Sample complexity upper bound for $\alpha>2$ via brute-force template matching}
\label{sec:upper-bound}

In this section we propose a recovery algorithm for the high SNR regime $\alpha>2$, which essentially matches our $\Omega(L/\log L)$ lower bound on the sample complexity. 
Our goal here is not to propose a new MRA algorithm, but rather to establish a matching upper bound on the \emph{statistical difficulty} of the problem; that is, we are studying the fundamental information-theoretic (rather than computational) limits of MRA.
\footnote{\revAdd{
		This distinction is not trivial in general. In the context of MRA, for instance, previous papers conjectured that a natural extension of the MRA model, called heterogeneous MRA, suffers from a fundamental computational-statistical gap~\cite{boumal2018heterogeneous,bandeira2017estimation}. We \emph{do not} claim, however, that such a computational-statistical gap holds for the MRA model considered in this paper, with $\alpha$ close to $2$.} 
}
\revDel{This is an important distinction because previous papers conjectured that a natural extension of the MRA model, called heterogeneous MRA, suffers from a fundamental computational-statistical gap~\cite{boumal2018heterogeneous,bandeira2017estimation}.}
In particular, the proposed algorithm is computationally intractable, and involves a brute-force search on an exponentially sized set of candidates. Moreover,  our approach is  tailored to the case $\alpha>2$, which is exactly the SNR regime where template matching is statistically possible. 

\paragraph{Outline of our algorithm}
Before diving into the technical details of our proposed scheme, we give a brief outline of the approach. 
The estimation algorithm works in two stages. Suppose we are given $n$ independent samples. We divide them into two subsamples of sizes $n_1$ and $n_2$, $n_1+n_2=n$. We do this so to ensure that the estimator $\widehat{Q}$ produced in step 1 is statistically independent of the additive noise in the  samples used for step 2. This simplifies our analysis considerably. The two stages performed by the algorithm are the following.
\begin{enumerate}
	\item \emph{Brute-force search for a template:} In the first stage, we use the first $n_1$ samples to find some direction $\widehat{Q}\in\Sphere^{L-1}$ (here $\Sphere^{L-1}$ is the unit sphere in $\RR^L$) such that $\widehat{Q}$ is sufficiently well-aligned with some shift of the true signal, that is, $\max_{\ell} L^{-1/2}\langle X,R_{\ell}^{-1}\widehat{Q}\rangle \ge 1-\eta$, where $\eta=\eta(\alpha)$ is small. To do this, we consider a fine-enough cover of the sphere, $\m{N}\subset \Sphere^{L-1}$, and take $\widehat{Q}\in \m{N}$ as the minimizer of a certain score:
	$\widehat{Q} = \argmin_{Q\in\m{N}} \sum_{i=1}^{n_1} s_i(Q)$,
	where $s_i(Q)$ is computed from the $i$-th sample $Y_i$. 
	Minimizing $\sum_{i=1}^{n_1} s_i(Q)$ over $\Sphere^{L-1}$ boils down to a brute-force search over the cover, whose size is exponential in $L$. Hence, this algorithm is not efficient.
	In principle, one could take at this point $\sqrt{L}\widehat{Q}\approx \|X\|\widehat{Q}$ as an estimator for~$X$. Unfortunately, the MSE of this estimator decays at a suboptimal rate with respect to the number of samples $n$; this is remedied by the second step.
	
	\item \emph{Alignment and averaging:} Using $\widehat{Q}$ from the previous step, we perform template matching on the remaining $n_2$ samples $Y_{1},\ldots,Y_{n_2}$ in order to estimate their shifts relative to $\widehat{Q}$:
	\[
	\widehat{R}_{\ell_i} = \argmax_{\ell} \langle Y_i,R_{\ell}\widehat{Q}\rangle . 
	\]
	The final estimator for $X$ is then the average of the aligned measurements:
	\[
	\widehat{X}=\frac{1}{n_2}\sum_{i=1}^{n_2} \widehat{R}_{\ell_{i}}^{-1}Y_{i}.
	\] 	 
\end{enumerate}
All the missing technical details are provided in the next two sections.
Due to space constraints, the proofs of all lemmas are given in the SI Appendix, Section~\ref{sec:proofs_sec_algorithm}.

\paragraph{Main result of this section.}
The main result of this section is the following:
\begin{proposition}\label{prop:upper-bound-informal}
	Suppose that $\alpha>2$, fix $\varepsilon>0$, and let $n,L\to \infty$. Then, there exists some $c(\alpha)>0$ depending on $\alpha$ such that if
	\[
	n_1 = c(\alpha)\sigma^2,\quad n_2 =  (1+o(1))\frac{\sigma^2}{\varepsilon}, 
	\]
	then the estimator $\widehat{X}$ returned by our algorithm satisfies $\rho(X,\widehat{X})\le \varepsilon$ with probability $1-o(1)$. 
\end{proposition}

Note that when $\varepsilon>0$ is small, the sample complexity is dominated by $n_2$:
\[
n =  c(\alpha)\sigma^2 + (1+o(1))\frac{\sigma^2}{\varepsilon} \approx (1+o(1))\frac{\sigma^2}{\varepsilon},
\]
and thus almost independent of the constant $c(\alpha)$.
Proposition~\ref{prop:upper-bound-informal} should be compared with the optimal achievable MSE for estimating a signal in AWGN, without the shifts $L^{-1}\Expt \|X-\widehat{X}_{\mathrm{MMSE}}\|^2=\frac{\sigma^2}{\sigma^2 + n}$.  
\paragraph{Proof of Theorem~\ref{thm:highsnrregime} (upper bound)} 
The upper bound for $\alpha>2$ follows readily from  Proposition~\ref{prop:upper-bound-informal}. 
To show this, we
construct a new estimator $[\widehat{X}]$ as follows: $[\widehat{X}]=\widehat{X}$ if $\|\widehat{X}\|\le 10\sqrt{L}$ and $[\widehat{X}]=0$ otherwise. 
Note that under the high-probability event $\|X\|\le 2\sqrt{L}$, necessarily 
$\rho(X,[\widehat{X}]) \le \rho(X,\widehat{X})$. Write
\begin{align*}
	\Expt \rho(X,[\widehat{X}]) = \Expt\left[\rho(X,[\widehat{X}])\Ind_{\|X\|\le 2\sqrt{L}}\right] + \Expt\left[ \rho(X,[\widehat{X}])\Ind_{\|X\|> 2\sqrt{L}} \right] .
\end{align*}
Under $\|X\|\le 2\sqrt{L}$, the random variable $\rho(X,[\widehat{X}])$ is bounded \revAdd{by a constant}, hence by Proposition~\ref{prop:upper-bound-informal},
\[
\Expt\left[\rho(X,[\widehat{X}])\Ind_{\|X\|\le 2\sqrt{L}}\right] \le \varepsilon+o(1)\,,
\]
\revAdd{since $\rho(X,\widehat{X})\le \varepsilon$ holds w.p. $1-o(1)$.}
As for the other term, 
\[
\Expt\left[ \rho(X,[\widehat{X}])\Ind_{\|X\|> 2\sqrt{L}} \right] \le \Expt\left[L^{-1/2}(\|X\|+10L^{1/2})\Ind_{\|X\|> 2\sqrt{L}}\right] \le 6\Expt\left[L^{-1/2}\|X\|\Ind_{L^{-1/2}\|X\|> 2}\right] \,,
\]
\revAdd{so that by Cauchy-Schwartz, 
	\[
	\Expt\left[L^{-1/2}\|X\|\Ind_{L^{-1/2}\|X\|> 2}\right] \le \left( L^{-1}\Expt[\|X\|^2] \right)^{1/2} \left( \Pr(\|X\|> 2\sqrt{L}) \right)^{1/2} = o(1) \,.
	\]}
Thus, $[\widehat{X}]$ uses $n=\left[ (1+o(1))/\varepsilon + c(\alpha) \right]\sigma^2$ samples and achieves $\Expt\rho(X,[\widehat{X}])\le \varepsilon+o(1)$, so that 
\[
\limsup_{L\to\infty} \frac{\SC(L,\alpha,\varepsilon)}{\sigma^2/\varepsilon}\le 1+O_\alpha(\varepsilon).
\]

\paragraph{Class of ``nice signals.''}
Before getting to the details of the algorithm, in the analysis that follows, it is convenient to treat the signal $X$ as fixed and belonging some class of ``nice'' signals. Specifically, we require that: (i) the signal is sufficiently uncorrelated with its shifts, in that $L^{-1}\langle X,R_{\ell}X\rangle \approx 0$ for all $\ell\ne 0$, and its norm is concentrated around $L^{-1}\|X\|^2\approx 1$; (ii) The Fourier (DFT) coefficients of $X$ are uniformly bounded.  

Let $f_0,\ldots,f_{L-1}\in \CC^L$ be the DFT basis vectors, that is, $(f_\ell)_j=L^{-1/2}e^{\frac{2\pi i}{L}\ell j}$, and $\m{F}\in U(L)$ be the matrix whose columns are $f_0,\ldots,f_{L-1}$, so that $\m{F}^*X\in \CC^{L}$ are the Fourier coefficients of $X$ (here $\m{F}^*$ denotes the Hermitian conjugate of $\m{F}$.) For $\kappa>0$, we formally consider the set 
\begin{equation}
	\signalClass_\kappa = \left\{ X\in\RR^L\quad:\quad \max_{\ell}\left| L^{-1}\langle X, R_{\ell}X\rangle - \Ind_{\{\ell=0\}}\right|\le \kappa,\quad\textrm{and }\|\m{F}^* X\|_{\infty} \le \sqrt{10\log(L)} \right\},
\end{equation} 
where $\Ind_{\{\ell=0\}}=1$ when $\ell=0$ and is zero otherwise.
We take $\kappa=o(1)$ sufficiently large so to ensure that when $X\sim \m{N}(0,I)$, the constraint \mbox{$\max_{\ell}\left| L^{-1}\langle X, R_{\ell}X\rangle - \Ind_{\{\ell=0\}}\right|\le \kappa$} holds with probability $1-o(1)$ as $L\to\infty$; by Lemma~\ref{lem:signal-is-nice}, we may choose $\kappa=c\log(L)/\sqrt{L}$ for $c>0$ a large enough constant. Let $\signalClass$ be the set corresponding to such choice. To lighten the notation, we will not keep track of $\kappa$ explicitly, instead referring to all vanishing terms as $o(1)$. For the other constraint, the exact bound $\|\m{F}^*X\|_\infty \le \sqrt{10\log(L)}$ is somewhat arbitrary, in that $10$ can be replaced with any constant greater than $4$. 
The following is quite immediate at this point:

\begin{lemma} \label{lem:XinX}
	Suppose that $X\sim \m{N}(0,I)$. Then, $\Pr(X\notin \signalClass)=o(1)$. 
\end{lemma}

We note that it is likely that without assuming that the estimation is over a class of ``nice'' signals (for example, the class $\signalClass_\kappa$), the situation changes. On that note, we mention the work~\cite{brunel2019learning}, where it is shown that there are signals $X$ for which the MLE only attains the rate $\rho(X,\widehat{X}_{\mathrm{MLE}})\sim n^{-1/2}$. 

\subsection{Step 1: Brute force template matching}\label{sect:brute-force}

Recall that our intermediate goal here is to find a direction $\widehat{Q}\in\Sphere^{L-1}$ such that $\max_{\ell}L^{-1/2}\langle X,R_{\ell}^{-1}\widehat{Q}\rangle \ge 1-\eta$, where $\eta>0$ is some desired accuracy level. 
Since, assuming $X\in \signalClass$, for any $Q\in\Sphere^{L-1}$, 
\[
\left\| \frac{X}{\|X\|}-R_{\ell}^{-1}Q \right\|^2 = 2 - 2\left\langle \frac{X}{\|X\|}, R_{\ell}^{-1}Q\right\rangle = 2 - 2L^{-1/2}\langle X, R_{\ell}^{-1}Q\rangle + o(1),
\]
then taking $\m{N}$ to be a $\sqrt{\eta}$-cover of $\Sphere^{L-1}$, it must contain some $Q \in \m{N}$ with $L^{-1/2}\langle Q,R_{\ell}^{-1} X\rangle \ge 1-\frac12 \eta + o(1)$.  It is well known that one can find a cover of the sphere which is not too large:
\begin{lemma}\label{lem:net-size}[Lemma 5.13 in \cite{van2014probability}]
	There exists an $\sqrt{\eta}$-cover $\m{N}$ of $\Sphere^{L-1}$ of size $|\m{N}|\le (3/\sqrt{\eta})^{L}$. That is, there exists a set $\m{N}\subset\Sphere^{L-1}$ of size $|\m{N}|\le (3/\sqrt{\eta})^{L}$, such that $\forall X\in\Sphere^{L-1}\,, \exists Q\in\m{N}$ with $\|X-Q\|\leq\sqrt{\eta}$.
\end{lemma}

For each $Q\in \m{N}$, we define its per-sample score:
\[
s_i(Q) = s^{\eta}_i(Q)=  \Ind\left[ \max_{\ell} L^{-1/2}\langle Y_i,R_{\ell}^{-1} Q \rangle \ge 1-\frac34 \eta \right], 
\]
and the total score $s(Q)=\sum_{i=1}^{n_1} s_i(Q)$, $n_1$ being the number of samples allocated for this step.
That is, $s(Q)$ is the number of samples $Y_i$ such that $L^{-1/2}\langle Q,R_{\ell}^{-1}Y_i\rangle  \ge 1-\frac34 \eta$ for some $\ell$. 
The returned estimator is then simply
\[
\widehat{Q} = \argmax_{Q\in \m{N}} s(Q) .
\]
Note that $s_i(\cdot)$ could be thought of as a discontinuous proxy for the log-likelihood (restricted to $X\in \Sphere^{L-1}$): $\log P(Y_i|X)= \log \sum_{\ell=0}^{L-1} \exp \left(  \frac{1}{\sigma^2}\langle X,R_{\ell}^{-1} Y_i\rangle  \right) + \mathrm{constant}$. When $\sigma$ is small, the log-likelihood is essentially dominated by $\max_{\ell} \sigma^{-2}\langle X, R_{\ell}^{-1} Y_i \rangle $. Maximizing the likelihood is computationally more straightforward (in the sense that this is a continuous optimization problem, no need to quantize the domain as we do); however, analyzing the MLE directly appears to be difficult~\cite{fan2020likelihood,katsevich2020likelihood}. 

We start by showing that there are only a few shifts $\ell$ such that $L^{-1/2}\langle X,R_{\ell}^{-1}Q\rangle$ are all large. 
\begin{lemma}\label{lem:not-many-large-shifts}
	Suppose that $X\in \signalClass$. For $Q\in \Sphere^{L-1}$, let 
	\[
	N_Q(h) = \left| \left\{ \ell \,:\, L^{-1/2}\left| \langle X,R_{\ell}^{-1} Q\rangle \right| \ge h \right\} \right| .
	\]
	Then, $N_Q(h) \le h^{-2}\|\m{F}^*X\|_\infty^2 \le h^{-2}\cdot 10\log(L)$.
\end{lemma}

We next show that if $\max_{\ell} L^{-1/2} \langle X, R^{-1}_\ell Q\rangle $ is small, then with high probability the score $s(Q)$ is not large. 

\begin{lemma}\label{lem:bad-has-small-score}
	Assume that $X\in\signalClass$, $\alpha>2$, $\eta<1-\sqrt{2/\alpha}$, and $L$ is large enough so that $\log(L) \le L^{3\eta^2\alpha/128}$. 	
	Suppose that $Q\in \Sphere^{L-1}$ is such that $\max_{\ell} L^{-1/2} \langle X, R_{\ell}^{-1} Q\rangle \le 1-\eta$, then 
	\[
	\Pr\left(s(Q) \ge n_1/2 \right) \le \left[ 16\left(2+\frac{640}{ \left( 1-\sqrt{\frac{2}{\alpha}}\right)^2 }\right)L^{-\eta^2\alpha/128} \right]^{n_1/2} .
	\]
\end{lemma} 

Next, we prove that if $\max_{\ell}\langle X,R_{\ell}^{-1}Q\rangle$ is sufficiently large, then $s(Q)$ is large with high probability.

\begin{lemma}\label{lem:good-has-large-score}
	Assume that $X \in \signalClass$, $\alpha>2$, and $L$ is large enough so that $L^{\eta^2\alpha/64}\ge 4$.  	
	Suppose that $Q \in \Sphere^{L-1}$ is such that $\max_{\ell}\langle X,R^{-1}_\ell Q\rangle \ge 1-5\eta/8$. Then,
	\[
	\Pr(s(Q)< n_1/2) \le e^{-n_1/32} .
	\] 
\end{lemma}

We are now ready to conclude the analysis of Step 1 of our algorithm.
\begin{proposition}\label{prop:brute-force-guarantee}
	Assume that $X\in \signalClass$, $\alpha>2$, and  $\eta<1-\sqrt{2/\alpha}$.
	Then, there is constant $c>0$, such that whenever 
	\[
	n_1\ge c\frac{L\log(1/\eta)}{\alpha \eta^2 \log(L)} = c\frac{\sigma^2\log(1/\eta)}{\eta^2},
	\]
	the vector $\widehat{Q}=\argmax_{Q\in \m{N}}s(Q)$ satisfies $\max_{\ell} \langle X,R_{\ell}^{-1}Q\rangle \ge 1-\eta$ with probability $1-o(1)$ as $n_1,L\to\infty$. In fact, the error probability decays exponenentially fast with $n_1$. 
\end{proposition}
\begin{proof}
	As argued in the beginning of this section, the $\sqrt{\eta}$-cover $\m{N}$ contains some $Q\in \Sphere^{L-1}$ such that $L^{-1/2}\langle X,R_{\ell}^{-1}Q\rangle \ge 1-\eta/2-o(1)\ge 1-5\eta/8$ for some $\ell$. By Lemma~\ref{lem:good-has-large-score}, with probability greater than $1-e^{-n_1/32}$, this vector has score $s(Q)\ge n_1/2$. It therefore suffices to show that with high probability, all the vectors $Q\in \m{N}$ that are bad, meaning that $\max_{\ell} L^{-1/2}\langle X, R_{\ell}^{-1}Q\rangle < 1-\eta$, have score $s(Q)<n_1/2$.
	By Lemmas~\ref{lem:net-size} and \ref{lem:bad-has-small-score},
	\begin{align*}
		\Pr\left(\exists \textrm{bad } Q\in\m{N} \,:\,s(Q)\ge n_1/2 \right) 
		&\le |\m{N}| \cdot \Pr\left(s(Q)\ge n_1/2\,\big|\,Q \textrm{ is bad}\right) \\
		&\le (9/\eta)^{L/2} \cdot \left[ 16\left(2+\frac{640}{ \left( 1-\sqrt{\frac{2}{\alpha}}\right)^2 }\right)L^{-\eta^2\alpha/128} \right]^{n_1/2} \\
		&\le \left( C(\alpha) e^{ -c_1\eta^2\alpha\log(L) + c_2 \frac{L}{n}\log(1/\eta)} \right)^{n_1} ,
	\end{align*} 
	where $c_1,c_2>0$ are absolute constants, and $C(\alpha)$ depends on $\alpha$. Then, this probability tends to $0$ as $n_1,L\to\infty$ (exponentially fast in $n_1$ ) whenever $n_1\ge c\frac{L\log(1/\eta)}{\alpha\eta^2\log(L)}$ for some other $c>0$.
\end{proof}

Note that at this point we could take $\widehat{X} = L^{1/2}\cdot \widehat{Q}$ as an estimator for $X$, so that 
\[
\rho(X,\widehat{X}) = \min_{\ell}\|L^{-1/2}X - R_\ell^{-1}Q\|^2 \le 2\eta + o(1),
\]
holds with high probability. For \emph{fixed} $\eta$, this estimator indeed captures the correct dimensional scaling of the sample complexity, namely, that $n=O(L/(\alpha\log L)$ samples are sufficient to get non-trivial alignment error. However, its dependence on $\eta$ is seemingly quite bad: for estimating a signal in AWGN, without the shifts, the optimal dependence on $\eta$ should look like $O(L/(\alpha\log L)\cdot \eta^{-1})$, rather than the much worse $O\left( L/(\alpha\log L) \cdot \eta^{-2}\log(1/\eta) \right)$ we were able to show. In the next section, we see how to achieve this ``correct''  rate by essentially recovering the shifts on all but a vanishing fraction of the samples, and averaging the properly aligned measurements.

\subsection{Step 2: Achieving optimal MSE decay rate by alignment and averaging}

Suppose that one has access to a known template $Q\in \Sphere^{L-1}$, such that $\langle X,Q\rangle \ge 1-\eta$.  Since $L^{-1}\|X\|^2=1+o(1)$, this is the same as having $\|L^{-1/2}X- Q\|^2 \le 2\eta+o(1)$, and since $\max_{\ell\neq 0}L^{-1}|\langle X,R_\ell X\rangle| = o(1)$, we see that for any $\ell\ne 0$, 
\[
\| L^{-1/2}R_\ell X - Q\| \ge \|L^{-1/2}[R_{\ell}X - X]\| - \|L^{-1/2}X-Q\| \ge \sqrt{2}-\sqrt{2\eta} - o(1).  
\]
In particular, we see that when $\sqrt{2\eta}<\sqrt{2}-\sqrt{2\eta}$, that is, $\eta<1/4$ (and $L$ is sufficiently large), there is a \emph{unique} $\ell$ (specifically, $\ell=0$) such that $\|L^{-1/2}X-R_{\ell}Q\|^2 \le 2\eta+o(1)$. In that case, the idea of matching a sample $Y_i=R_{\ell_i}X + \sigma Z$ against the template $Q$ becomes well-posed, in the sense that its desired outcome is clear: we would like to recover the shift $R_{\ell_i}$. 

\begin{lemma}\label{lem:match-with-partial-template}
	Assume that $X\in\signalClass$ and $\alpha>2$. 
	Let $Y=R_{\ell}X+\sigma Z$, and suppose that $Q\in \Sphere^{L-1}$ is independent of $Y$ and satisfies $\max_{\ell'} L^{-1/2}\langle X, R_{\ell'}^{-1}Q \rangle \ge 1-\eta$, where 
	\[
	\sqrt{\eta} < \frac{1}{2}(1-\sqrt{2/\alpha}) .
	\]
	Denote the maximizing shift by $\ell^*$.
	Let $\widehat{\ell} = \argmax_{\ell'} \langle Y, R_{\ell'} Q\rangle$. Then 
	\[
	\Pr\left(\widehat{\ell}\ne \ell-\ell^* \right) \le 2  L^{-\frac12 \alpha\left(1/2-1/\sqrt{2\alpha} - \sqrt{\eta}\right)^2+o(1)}.
	\]
\end{lemma}  

Given Lemma~\ref{lem:match-with-partial-template}, we propose the following estimation strategy. 
Suppose we would like to estimate~$X$ up to error $\rho(X,\widehat{X})\le \varepsilon<1$. Fix some $\eta>0$ with $\sqrt{\eta}<(1-\sqrt{2/\alpha})/2$ (for concreteness, say $\eta=(1-\sqrt{2/\alpha})^2/16$). We first apply the algorithm of Step 1 (Setion~\ref{sect:brute-force}) to obtain $\widehat{Q}\in\Sphere^{L-1}$ such that $\max_{\ell}\langle X,R_{\ell}^{-1}\widehat{Q}\rangle \ge 1-\eta$. Assuming that $n_1\ge \frac{c\log(1/\eta)}{\eta^2}\sigma^2=c_{\eta}\sigma^2$, we are successful with probability $1-o(1)$. Let $\ell^*$ be such that $\langle X,R^{-1}_{\ell^*}Q\rangle \ge 1-\eta$. Next, for $n_2$ new independent samples, we compute for each measurement $\widehat{\ell}_i = \argmax_{\ell}\langle Y_i,R_{\ell}\widehat{Q}\rangle$ and return the aligned sample average:
\begin{equation}\label{eq:aligned-sample-avg}
	\widehat{X} = \frac{1}{n_2}\sum_{i=1}^{n_2} R_{\widehat{\ell}_i}^{-1} Y_i.
\end{equation} 
Lemma~\ref{lem:match-with-partial-template} tells us that we should expect most of the aligned measurements $R_{\widehat{\ell}_i}^{-1} Y_i$ to be well-aligned with $R_{\ell^*}X$, that is, $R_{\widehat{\ell}_i}^{-1} Y_i = R_{\ell^*}X + \m{N}(0,\sigma^2 I)$. This means that, $\widehat{X}\approx R_{\ell^*}X + \m{N}(0,(\sigma^2/n_2)I)$, hence $\rho(X,\widehat{X})\le L^{-1}\|R_{\ell^*}X-\widehat{X}\|^2\approx \sigma^2/n_2$, which is smaller than $\varepsilon$ if $n_2\ge \sigma^2/\varepsilon$. We make this argument precise below:

\begin{proposition}\label{prop:averaging-guarantee}
	Assume that $X\in\signalClass$ and $\alpha>2$. Fix $\varepsilon>0$ and some $\eta<\frac12(1-\sqrt{2/\alpha})^2$. Let $\widehat{Q}\in \Sphere^{L-1}$ be the output of Step 1 (run with a tuning parameter $\eta$ and $n_1$ samples).   
	Let $\widehat{X}$ be as in equation~\eqref{eq:aligned-sample-avg}, computed from $n_2$ new samples.	
	Suppose that $n_1,n_2,L\to \infty$ with 
	\[
	n_1/\sigma^2 \to \gamma_1,\quad n_2/\sigma^2\to \frac{\gamma_2}{\varepsilon},
	\]
	where $\gamma_1$ and $\gamma_2$ are constants satisfying
	\[
	\gamma_1 = \gamma_1(\eta) \ge \frac{c\log(1/\eta)}{\eta^2},\quad \gamma_2>1,
	\]
	($c$ being the universal constant from Proposition~\ref{prop:brute-force-guarantee}).
	Then,
	\[
	\Pr\left( \rho(X,\widehat{X})\le \varepsilon \right) \to 1 .
	\]
\end{proposition}

Proposition~\ref{prop:upper-bound-informal} now immediately follows from Lemma~\ref{lem:XinX} and Proposition~\ref{prop:averaging-guarantee}.

\section{Conclusions and extensions}
\label{sec:conclusion}

In this work we have studied the sample complexity of the MRA problem in the limit of large $L$. In this regime, we have shown that the parameter $\alpha=\frac{\sigma^2\log{L}}{L}$ plays a crucial role in characterizing the best attainable performance of any estimator.

As mentioned above, the MRA model is primarily  motivated by the cryo-EM technology to constitute the 3-D structure of biological molecules.
In the cryo-EM literature, it was shown that it is effective to assume that the molecule was drawn from a  Gaussian prior  with decaying power spectrum~\cite{scheres2012relion}. In addition, the 3-D rotations are usually not distributed uniformly over the group $SO(3)$. We now discuss briefly how these  different aspects can  be potentially incorporated into our framework.

\paragraph{Prior on the signal} Our model assumes a Gaussian i.i.d. prior on the signal $X$ to be reconstructed. While this assumption lends itself to a relatively clean analysis, and allows to compare our bounds on $\SC(L,\alpha,\varepsilon)$ to the simple benchmark $n^*_{\text{AWGN}}(L,\alpha,\varepsilon)$, many of our results can be generalized to treat other priors on $X$. 
In particular, all of our sample complexity lower bounds are based on lower bounding the mutual information between $X$ and $\hat{X}$ under the constraint $\mathbb{E}[\rho(X,\hat{X})]\leq \varepsilon$ on the one hand, and upper bounding $I(X;Y^n)$ under the MRA model, on the other hand. In Proposition~\ref{prop:mrardf} we have relied on the Gaussian rate distortion function to lower bound $I(X;\hat{X})$ for any estimator that achieves MSE at most $\varepsilon$. For $X$ whose distribution is not $\m{N}(0,I)$, we can either compute the corresponding rate distortion function explicitly, or simply apply Shannon's lower bound $R(D)\geq h(X)-\frac{L}{2}\log(2\pi e D)$, see~\cite{berger71}. Our upper bounds on $I(X;Y^n)$ in the regime $\alpha>1$ are based on Lemma~\ref{lem:gaussian-MI}, followed by lower bounding $I(R^n;X|Y^n)$ using Fano-like arguments. It is easy to see that~\eqref{eq:alphageq1MIbound} continues to hold, with $\leq$ instead of $=$, for any random variable~$X$ with $\mathbb{E}\|X\|^2\leq L$. Furthermore, the lower bounds on $I(R^n;X|Y^n)$ we derive in Section~\ref{subsec:MIviaTM} {remain valid whenever $\frac{\|X\|}{L}$ is sufficiently concentrated around $1$ and $\frac{\langle X,R_{\ell}X\rangle}{L}$ is sufficiently concentrated around $0$ for all $\ell=1,\ldots,L-1$. In particular, this is the case for (sufficiently light-tailed) i.i.d. zero-mean and unit variance distributions.} In light of the discussion above, we see that the parameter $\alpha=\frac{\sigma^2\log{L}}{L}$ is of great importance whenever the random signal $X$ {satisfies the above concentration requirements and has} differential entropy proportional to $L$.

\paragraph{Shift distribution} 
Assuming uniform prior on the i.i.d.\ shifts $R_{\ell_1},\ldots,R_{\ell_n}$ is a worst-case analysis. Indeed, for any given distribution, shifting all measurements again $R_{u_i}Y_i$, for {$u_i\stackrel{i.i.d.}{\sim}\Unif(\{0,\ldots,L-1\})$}  before feeding them to the estimator leads to~\eqref{eq:model}. 
However, previous works (for fixed $L$) showed that harnessing non-uniformity can make a big difference in the sample complexity~\cite{abbe2018multireference,sharon2020method}. 
With some effort, our upper bounds on $I(X;Y^n)$ in the regime $\alpha>1$ should also extend to treat this case. Here, the main challenge is to generalize Lemma~\ref{lem:not-many-candidiates} to the case of non-uniform distribution, i.e., to find a sharp estimate on the smallest possible size of a list of candidates for the true shift, which contains the true shift with high probability. 

\paragraph{Extension to other groups} 
We believe that many aspects of our information-theoretical analysis can be generalized to other (families of) discrete groups, denoted here by $\m{G}_L$, 
which satisfy the following properties (roughly speaking): (i) If $X$ is suitably generic and $g\ne h$, then $\langle gX, hX\rangle$ is very small - concretely, if $X\sim \m{N}(0,I)$, then $\Expt [\langle gX,hX\rangle]=0$; (ii) The size of the group~$|\m{G}_L|$ does not grow too fast (strictly less than exponentially fast in $L$). These conditions imply that whenever $X$ is isotropic and sufficiently light-tailed (e.g., sub-Gaussian), $\{gX\}_{g\in \m{G}}$ are ``almost orthogonal.'' The proper noise scaling to consider would then be ${\sigma^2 = \frac{L}{\alpha \log|\m{G}_L|}}$, with $\alpha=2$ being the critical noise level---this comes from the fact that ${\max_{g\in \m{G}_L} \langle gX,Z\rangle \approx \sqrt{2\log|\m{G}_L|}}$. 
For continuous compact groups , we suspect that one might be able to apply some of our arguments by cleverly discretizing the suitable group action.
Carrying out a program of this type seems as a promising direction for future research.

\section*{Acknowledgment}
E.R. and O.O. are supported in part by the ISF under Grant 1791/17. E.R. is supported in part by an Einstein-Kaye fellowship from the Hebrew University of Jerusalem. T.B. is supported in part by NSF-BSF grant no. 2019752, and  the Zimin Institute for Engineering Solutions Advancing Better Lives.

\bibliographystyle{plain}
\bibliography{ref_thesis}


\newpage
\appendix 

\section{Information Theoretic Background}\label{sec:ITback}

In this section we review some basic information theoretic definitions and results that are needed throughout this paper. The proofs of the results below can be found in any textbook on information theory, e.g.~\cite{cover2012elements}, and are therefore omitted.

For a discrete random variable $X\sim P_X$ supported on the alphabet $\m{X}$, the entropy is defined as
\begin{align*}
	H(X)=H(P_X):=\sum_{x\in\m{X}} P_X(x)\log\frac{1}{P_X(x)}=\Expt_{X\sim P_X}\left[\log\frac{1}{P_X(X)}\right].
\end{align*}
For a pair of random variables $(X,Y)\sim P_{XY}$, where $X$ is discrete, the conditional entropy of $X$ given $Y$ is defined as
\begin{align*}
	H(X|Y):=\Expt_{y\sim P_Y}\left[H(X|Y=y)\right]=\Expt_{y\sim P_Y}\left[H(P_{X|Y=y})\right].
\end{align*}

Similarly, if $X$ is a continuous random variable on $\RR^d$ with density $p_X$, its differential entropy is defined as
\begin{align*}
	h(X)=h(P_X):=\int_{x\in\RR^d} p_X(x)\log\frac{1}{p_X(x)} dx=\Expt_{X\sim P_X}\left[\log\frac{1}{p_X(X)}\right].
\end{align*}
For a pair of random variables $(X,Y)\sim P_{XY}$, where $X$ is continuous and has conditional density $p_{X|Y=y}$ for all $y\in\m{Y}$, where $\m{Y}$ is the alphabet of $Y$, the conditional entropy is defined as
\begin{align*}
	h(X|Y)=\Expt_{y\sim P_Y}\left[h(X|Y=y)\right]=\Expt_{y\sim P_Y}\left[h(P_{X|Y=y})\right].
\end{align*}

\begin{proposition}[Properties of entropy and differential entropy] \leavevmode
	
	\begin{enumerate}
		\item \textbf{Non-negativity of entropy:} For a discrete random variable $X$ the entropy satisfies $H(X)\geq 0$, with equality if and only if  $X$ is deterministic.\label{entropy:noneg}
		\item \textbf{Uniform distribution maximizes entropy:} 	For a discrete random variable $X$ supported on $\m{X}$
		\begin{align*}
			H(X)\leq\log|\m{X}|,
		\end{align*}
		and this is attained with equality if and only if  $X\sim\Unif(\m{X})$.\label{entropy:unif}
		
		\item \textbf{Gaussian distribution maximizes differential entropy under second moment constraints:} Suppose that the continuous random variable $X$ is  supported on $\RR^d$, and has covariance matrix $\Sigma=\Expt[(X-\Expt(X))(X-\Expt(X))^\T]$. Then, 
		\begin{align}
			h(X)\leq \frac{1}{2}\log\left((2\pi e)^d \det(\Sigma)\right),
		\end{align}
		and this is attained with equality if and only if $X\sim \m{N}(\mu,\Sigma)$ for some $\mu\in\RR^d$.\label{entropy:gauss}
		
		\item \textbf{Chain rule:} For discrete random variables $(X,Y)\sim P_{XY}$ we have
		\begin{align*}
			H(X,Y)=H(X)+H(Y|X)=H(Y)+H(X|Y).
		\end{align*} 
		For continuous random variables $(X,Y)\sim P_{XY}$, we have
		\begin{align*}
			h(X,Y)=h(X)+h(Y|X)=h(Y)+h(X|Y).
		\end{align*} \label{entropy:chainrule}
		
		\item \textbf{Concavity:} The functions $P_X\mapsto H(P_X)$ and $P_X\mapsto h(P_X)$ are concave. Consequently, conditioning reduces entropy, that is
		\begin{align*}
			H(X|Y)\leq H(X)
		\end{align*}
		if $X$ is discrete, and 
		\begin{align*}
			h(X|Y)\leq h(X)
		\end{align*}
		if $X$ is continuous. In both cases, the bounds are attained with equality iff $X$ and $Y$ are statistically independent. \label{entropy:concavity}
	\end{enumerate}
	\label{prop:entropy}
\end{proposition}

We will also make use of Fano's inequality, 
as stated below.
\begin{proposition}[Fano's inequality]
	Let $(X,Y)\sim P_{XY}$, where $X$ is a discrete random variable supported on $\m{X}$. Then, for any estimator $\hat{X}=\hat{X}(Y)$ of $X$ from $Y$, we have
	\begin{align*}
		H(X|Y)\leq \log{2}+\Pr(X\neq \hat{X})\log|\m{X}|.
	\end{align*}
	\label{prop:fano}
\end{proposition}


If both $(X,Y)\sim P_{XY}$ are discrete, the mutual information between $X$ and $Y$ is defined as
\begin{align*}
	I(X;Y)=H(X)-H(X|Y)=H(Y)-H(Y|X),
\end{align*}
and if they are both continuous 
\begin{align*}
	I(X;Y)=h(X)-h(X|Y)=h(Y)-h(Y|X).
\end{align*}
If one is discrete, say $X$, and the other continuous, say $Y$, then
\begin{align*}
	I(X;Y)=H(X)-H(X|Y)=h(Y)-h(Y|X).
\end{align*}
For a triplet of random variables $(X,Y,Z)\sim P_{XYZ}$, the conditional mutual information is defined as
\begin{align*}
	I(X;Y|Z)=\Expt_{z\sim P_Z}\left[I(X;Y|Z=z)\right],
\end{align*}
where $I(X;Y|Z=z)$ is the mutual information between $X$ and $Y$ under the distribution $(X,Y)\sim P_{XY|Z=z}$.

\begin{proposition}[Properties of Mutual Information]
	\leavevmode
	\begin{enumerate}
		\item \textbf{Non-negativity of mutual information:} $I(X;Y)\geq 0$ with equality iff $X$ and $Y$ are statistically independent.\label{mi:dperp}
		\item \textbf{Chain rule:} For $(X,Y,Z)\sim P_{XYZ}$ we have
		\begin{align*}
			I(X;Y,Z)=I(X;Y)+I(X;Z|Y)=I(X;Z)+I(X;Y|Z).
		\end{align*}\label{mi:chainrule}
		\item \textbf{Data processing inequality:} Assume $X-Y-Z$ is a Markov chain in this order, that is their joint distribution decomposes as $P_{XYZ}=P_{X} P_{Y|X}P_{Z|Y}$, then
		\begin{align*}
			I(X;Z)\leq I(X;Y).
		\end{align*}\label{mi:dpi}
		\item \textbf{Invertible functions:} For any function $f:\m{Y}\to\m{A}$, where $\m{A}$ is an arbitrary alphabet, we have $I(X;f(Y))\leq I(X;Y)$ with equality if $f$ is invertible.\label{mi:invertible}
		\item \textbf{Mutual information for memoryless channels:} Let $(X^n,Y^n)\sim P_{X^n Y^n}=P_{X^n}P_{Y^n|X^n}$ and assume the channel from $X^n$ to $Y^n$ is a product channel, that is $P_{Y^n|X^n}=\prod_{i=1}^n P_{Y_i|X_i}$. Then
		\begin{align*}
			I(X^n;Y^n)\leq \sum_{i=1}^n I(X_i;Y_i).
		\end{align*}
		This bound is attained with equality if $P_{X^n}=\prod_{i=1}^n P_{X_i}$, i.e., if $X^n$ is memoryless as well.\label{mi:memoryless}
		\item\label{mi:Gaussian} \textbf{Gaussian mutual information:} Let $X,Z\sim\m{N}(0,I)$ be statistically independent $L$-dimensional random vectors with i.i.d. standard normal entries. Then
		\begin{align*}
			I(X;X+\sigma Z)=\frac{L}{2}\log\left(1+\frac{1}{\sigma^2}\right).
		\end{align*}
	\end{enumerate}
	\label{prop:MIproperties}
\end{proposition}

For a random variable $X\sim P_X$ supported on alphabet $\m{X}$, a reconstruction alphabet $\hat{\m{X}}$ and a distortion measure $d:\m{X}\times \hat{\m{X}}\to\RR$, the rate distortion function (RDF) is defined as
\begin{align*}
	R(D)=\min_{P_{\hat{X}|X} \ : \ \mathbb{E}[d(X,\hat{X})]\leq D  }I(X;\hat{X}),
\end{align*}
where both $I(X;\hat{X})$ and $\mathbb{E}[d(X,\hat{X})]$ are evaluated with respect to the joint distribution $P_X P_{\hat{X}|X}$. The solution of the optimization problem above for the quadratic Gaussian case is well known, and is summarized in the proposition below.
\begin{proposition}[Quadratic Gaussian RDF]
	\leavevmode
	Let $X\sim\m{N}(0,\sigma^2 I)$ be a random vector in $\RR^L$, $\hat{\m{X}}=\RR^L$, and $d(x,\hat{x})=\frac{1}{L}\|x-\hat{x}\|^2$. Then,
	\begin{align}
		R(D)=\frac{L}{2}\log\left(\frac{\sigma^2}{D}\right).
	\end{align}
	In particular, if $X\sim\m{N}(0,\sigma^2 I)$ and $\hat{X}$ is such that $\frac{1}{L}\mathbb{E}\|X-\hat{X}\|^2\leq D$, then
	\begin{align*}
		I(X;\hat{X})\geq \frac{L}{2}\log\left(\frac{\sigma^2}{D}\right).
	\end{align*}
	\label{prop:GaussianRDF}
\end{proposition}

\section{Some remarks on the capacity of the MRA channel}
\label{sec:capacity}

One can think of the model $Y=RX + \sigma Z$ as a communication channel whose input is $X$ and output is $Y$. A natural question in information theory, then, is to find the \emph{capacity} of this channel, defined as
\[
C_{\mathrm{MRA}}(L,\sigma^2) = \max_{P_X\,:\,\Expt\|X\|^2 \le L} I(X;Y) ,
\]
where the optimization is over all input distributions $X$ obeying a mean power constraint $\Expt\|X\|^2 \le L$. The channel capacity is a central quantity in information theory, and characterizes exactly the fundamental limits of data transmission over this channel: in each channel use, one can at best trasmit reliably $C_{\mathrm{MRA}}$ nats of information. 

Determining the capacity of the additive white Gaussian channel $Y=X+\sigma Z$ is a classical problem. It is well-known that 
\[
C_{\mathrm{AWGN}}(L,\sigma^2) = \frac{L}{2}\log(1+\sigma^{-2}),
\]
and the capacity-achieving distribution is i.i.d Gaussian $X\sim \m{N}(0,I)$. It is easy to see that $C_{\mathrm{MRA}} \le C_{\mathrm{AWGN}}$. Indeed, note that $Y=RX+\sigma Z \overset{d}{=} R(X+\sigma Z)$ (by rotation invariance), hence by the data processing inequality (Proposition~\ref{prop:MIproperties} item \ref{mi:dpi}), applied to the Markov chain $X - (X+\sigma Z) - R(X+\sigma Z)$, we get 
\[
I(X;X+\sigma Z) \ge I(X;R(X+\sigma Z)) = I(X;Y),
\]
from which $C_{\mathrm{AWGN}} \ge C_{\mathrm{MRA}}$ follows. 
At this point, one naturally wonders: (i) Can something non-trivial be said about the ratio $C_{\mathrm{MRA}}/C_{\mathrm{AWGN}}$; in particular, when is it approximately one  (say as $L,\sigma^2\to 
\infty$)? (ii) What is the capacity achieving input distribution for the MRA channel? In particular, is $X\sim \m{N}(0,I)$ the capacity achieving input distribution at some (every?) SNR regime?

At very high SNR, namely $\sigma^{-2} L = \omega(\log(L))$, Eq.~\eqref{eq:MI-entropy-gap} tells us that an i.i.d Gaussian input is ``essentially'' capacity achieving: if $X\sim \m{N}(0,I)$, then 
\[
I(X;Y) \ge C_{\mathrm{AWGN}} - \log(L) = \frac{L}{2}\log(1+\sigma^{-2}) - \log(L),
\]
and the loss of information, $\log(L)$ nats, is negligible compared to $\frac{L}{2}\log(1+\sigma^{-2})$.

At very low SNR, however, it turns out that an i.i.d input distribution is very much suboptimal. Consider the input distribution $X\sim \m{N}(0, {\bf 1}{\bf 1}^\T)$, that is, we allocate the entire power budget on the direction ${\bf 1}/\sqrt{L} = (1/\sqrt{L},\ldots,1/\sqrt{L})$. Since all the coordinates of $X$ are the same, the signal is completely invariant to the shifts, meaning that $X=RX$ exactly. In that case,
\[
I(X;Y) = I(X;X+\sigma Z) = \frac12 \log(1+ \sigma^{-2}L), 
\]
so that under \emph{extremely low SNR}, where $\sigma^{-2}L<1$ is a constant but small number, we have $I(X;Y) = \frac12 \sigma^{-2}L - O(\sigma^{-4}L^2)$. We can also expand $C_{\mathrm{AWGN}} = \frac12 L\sigma^{-2} + O(\sigma^{-4}L)$, so that $I(X;Y)$ matches $C_{\mathrm{AWGN}}$ to leading order in the SNR. On the other hand, recall that for an i.i.d input distribution $X\sim \m{N}(0,I)$, we have seen that if the SNR is $\sigma^{-2}L<\log(L)$ then already $I(X;Y)=o(1)$. Thus, i.i.d inputs are highly suboptimal at low SNR.  

Determining the channel capacity and the capacity-achieving input distribution, inbetween the extreme SNR regimes $\sigma^{-2}L = \omega(\log(L))$ and $\sigma^{-2}L=o(1)$, looks like an interesting but quite challenging task. An i.i.d input $\m{N}(0,I)$ has the advatange that it utilizes optimally the available degrees of freedom ($L$, the dimension); its disadvantage is that it does not play well with  the random shift, in that the signals $R_{\ell} X$ are very different to one another. On the other hand, the input $\m{N}(0,{\bf 1}{\bf 1}^\T )$ mitigates best the negative effect of the random shift (it is not affected by it at all), but this is done at the expense of the available degrees of freedom (one instead of $L$). 
It is interesting to find out how the capacity achieving distribution balances delicately between these two effects.

\section{Proof of Lemma~\ref{lem:signal-is-nice}}
\label{sec:proof_lemma_signal-is-nice}

Before getting to the proof, we recall the Hanson-Wright inequality:

\begin{lemma}[Hanson-Wright inequality for sub-Gaussian random vectors, Theorem 1.1 in \cite{rudelson2013hanson}]\label{lem:hanson-wright}
	Let $X$ be a random vector with independent entries such that for all $i$,
	\[
	\Expt X_i = 0,\quad \|X_i\|_{\psi_2} \le K ,
	\]
	where $\|X_i\|_{\psi_2} = \inf \left\{ s>0\,:\,\Expt e^{(X_i/s)^2} \le 2 \right\}$. Let $A$ be any matrix.
	Then,	there is a universal constant $c>0$ such that 
	\begin{align*}
		\Pr\left(\left| X^\T A X - \Expt(X^\T A X) \right| > t \right) \le 
		2\exp\left[ -c\min\left(\frac{t^2}{K^4\|A\|_F^2}, \frac{t}{K^2\|A\|}\right) \right] .
	\end{align*}
\end{lemma}

It is immediate to verify that if $X\sim \m{N}(0,\sigma^2)$, then $\|X\|_{\psi_2} = \sigma/\sqrt{2\log 2} = c\sigma $. Also,
for any $\ell$, $\|R_{\ell}\|=1$ (since $R_\ell\in O(L)$) and therefore $\|R_{\ell}\|_F^2 \le L$. Also,
\[
\Expt( \langle X,R_{\ell}X\rangle ) = \trace(R_{\ell})=  \begin{cases}
	L \quad&\textrm{ if } \ell = 0, \\
	0 \quad&\textrm{ otherwise.}
\end{cases} 
\]
By the Hanson-Wright inequality, Lemma~\ref{lem:hanson-wright}, 
\begin{align*}
	\Pr\left(\left| \langle X,R_{\ell}X\rangle -\Expt\left( \langle X,R_{\ell}X\rangle\right)  \right| \ge L\kappa \right) \le 
	2\exp \left( -c\min((L\kappa)^2/L, L\kappa) \right) = 2\exp\left( -cL\min(\kappa,\kappa^2)\right).
\end{align*}
The claimed result follows by a union bound. 

\section{The spectrum of the operators $R_{\ell}$}
\label{sec:spectrum_shift_operator}

We recall some elementary facts about the spectrum of the operators $R_{\ell}$:
\begin{lemma}\label{lem:DFT-basis}
	The eigenvalues of $R_{\ell}+R_{\ell}^\T$ are exactly (with mutliplicities) $\lambda_{\ell,k} = 2\cos\left(\frac{2\pi}{L}\ell k\right)$, $k=0,\ldots,L-1$. Moreover,
	\[
	\sum_{k=0}^{L-1} \lambda_{\ell,k} = \begin{cases}
		L \quad&\textrm{ if } \ell = 0,\\
		0 \quad&\textrm{ otherwise.}
	\end{cases}
	\] 
\end{lemma}
\begin{proof}
	Let $f_k\in\CC^L$, $k=0,\ldots,L-1$, be the DFT basis vectors, namely  $f_{k,j} = L^{-1/2}e^{\frac{2\pi i}{L} k j}$. It is immediate to verify that $f_{k}$ is an eigenvector of $R_{\ell}$ with eigenvalue $\lambda_{\ell,k}=e^{\frac{2\pi i}{L}\ell k}$:
	\[
	(R_\ell f_k)_j = (f_{k})_{j+\ell} = e^{\frac{2\pi i}{L}\ell k }(f_{k})_j .
	\]
	Hence, $(R_{\ell} + R_{\ell}^\T)f_{k} = (e^{\frac{2\pi i}{L}\ell k }+e^{-\frac{2\pi i}{L}\ell k })f_{k} = 2\cos\left(\frac{2\pi i}{L}\ell k\right) f_k$. This means that $\lambda_{\ell,k}$ are the eigenvalues of $R_{\ell}+R_{\ell}^\T$ as an operator $\CC^L \to \CC^L$. But since $R_{\ell}+R_{\ell}^\T$ is also diagonalizable over $\RR^L$ by an orthogonal matrix, there also exists a \emph{real} orthonormal eigenbasis $u_1,\ldots,u_L\in \RR^L$ with $(R_\ell + R_{\ell}^\T)u_k = \lambda_{\ell,k}u_k$ . 
	As for the last claim, it follows from $\sum_{k=0}^{L-1}\lambda_{\ell,k} = 2\Re \left\{ \sum_{k=0}^{L-1} e^{\frac{2\pi i}{L}\ell k}\right\}$, the right-hand side being $L$ when $\ell=0$ and zero otherwise. 
\end{proof}

\section{Proof of Lemma~\ref{lem:not-many-candidiates}}
\label{sec:proof_lem_not_many_candidates}
Suppose that the event $\m{A}=\m{A}(\kappa)$ from Lemma~\ref{lem:signal-is-nice} holds, meaning that $\left| L^{-1}\|X\|^2-1\right|\le \kappa$ and ${\max_{\ell'\ne 0} L^{-1}\left| \langle X,R_{\ell'}X \rangle \right| \le \kappa}$. 
Observe that 
\[
R\notin \m{S}_\tau \Leftrightarrow \frac{\sigma \langle X,R^{-1}Z\rangle }{\|X\|^2} < -\tau .
\]
Conditioned on $X$, 
\[
\frac{\sigma \langle X,R^{-1}Z\rangle }{\|X\|^2} \sim \m{N}\left(0, \sigma^2/\|X\|^2 \right),
\]
and under $\m{A}$, this variance is $\sigma^2/\|X\|^2 = \frac{L}{\|X\|^2\cdot\alpha \log(L)}\le \frac{1}{\alpha(1-\kappa) \log(L)}$. Thus,
\[
\Pr\left(R\notin \m{S}_\tau \,\big|\, \m{A}\right) \le e^{-\frac12 \tau^2 \alpha(1-\kappa) \log(L) } = L^{-\frac12 \tau^2 \alpha(1-\kappa)} .	
\]
Now, suppose that $R'\ne R$. Then
\begin{align*}
	\Pr\left(R'\in \m{S}_\tau \,\big|\, \m{A},R\right) 
	&= \Pr\left( \frac{\langle X,(R')^{-1}R X\rangle }{\|X\|^2} + \frac{\sigma \langle X, (R')^{-1}Z \rangle }{\|X\|^2} \ge 1-\tau \, \Big|\, \m{A}, R \right) \\
	&\le \Pr\left( \frac{\sigma \langle X, (R')^{-1}Z \rangle }{\|X\|^2} \ge 1-\tau-\frac{\kappa}{1-\kappa}  \,\Big|\, \m{A}, R \right) \\
	&\le L^{-\frac12 \alpha(1-\kappa)\left(1-\tau-\frac{\kappa}{1-\kappa}\right)^2},
\end{align*}
where we used the fact that under $\m{A}$, $\frac{\langle X,(R')^{-1}R X\rangle }{\|X\|^2}\le  \kappa/(1-\kappa)$, and uniformly bounded the variance of $\frac{\sigma \langle X, (R')^{-1}Z \rangle }{\|X\|^2}$ conditioned on $X$ and under $\m{A}$ as before. Since the bound above is uniform in $R$, of course, 
\[
\Pr\left(R'\in \m{S}_\tau \,\big|\, \m{A}\right) \le L^{-\frac12 \alpha(1-\kappa)\left(1-\tau-\frac{\kappa}{1-\kappa}\right)^2}.
\]
Now,
\begin{align*}
	\Expt\left[|\m{S}_\tau| \,\Big|\, \m{A}\right] \le 1 + (L-1)\cdot L^{-\frac12 \alpha(1-\kappa)\left(1-\tau-\frac{\kappa}{1-\kappa}\right)^2} \le 1 + L^{1-\frac12 \alpha(1-\kappa)\left(1-\tau-\frac{\kappa}{1-\kappa}\right)^2} .
\end{align*}
Setting $M=L^{1-\frac12 \alpha(1-\kappa)\left(1-\tau-\frac{\kappa}{1-\kappa}\right)^2 + \zeta}$, by Markov's inequality, and assuming $\alpha\le 2$, 
\begin{align*}
	\Pr\left( |\m{S}_\tau|\ge M \, \big|\,\m{A} \right) \le \frac{1 + L^{1-\frac12 \alpha(1-\kappa)\left(1-\tau-\frac{\kappa}{1-\kappa}\right)^2 }}{L^{1-\frac12 \alpha(1-\kappa)\left(1-\tau-\frac{\kappa}{1-\kappa}\right)^2 + \zeta}} \le 2L^{-\zeta} .
\end{align*}
Combining both estimates and taking a union bound,
\begin{align*}
	\Pr\left( R\notin \m{S}_\tau \textrm{ or } |\m{S}_\tau|> M \right) 
	&\le 	\Pr\left(R\notin \m{S}_\tau \,\big|\, \m{A}\right) + \Pr\left( |\m{S}_\tau|\ge M \,\big|\,\m{A} \right) + \Pr\left(\overline{\m{A}}\right) \\
	&\le L^{-\frac12 \alpha(1-\kappa)\left(1-\tau-\frac{\kappa}{1-\kappa}\right)^2} + 2L^{-\zeta} + \Pr\left(\overline{\m{A}}\right) \\
	&\le L^{-\frac12 \alpha(1-\kappa)\left(1-\tau-\frac{\kappa}{1-\kappa}\right)^2} + 2L^{-\zeta} + 2Le^{-c L\min(\kappa,\kappa^2) },
\end{align*}
where the last inequality follows from Lemma~\ref{lem:signal-is-nice}.

\section{Proof of Theorem~\ref{thm:pmra} (Projected MRA)}\label{sect:pmra-proof}

In this section, we sketch a proof of Theorem~\ref{thm:pmra}. Recall that in the PMRA model, the measurements $Y_1,\ldots,Y_n\in \RR^{L'}$ have the form
\[
Y_i = \pi_S R_{\ell_i} X + \sigma {Z_i},
\]
where {$X\sim \m{N}(0,I)$} {is $L$-dimensional, $Z_i\stackrel{i.i.d.}{\sim}N(0,I)$ are $L'$-dimensional}, and $\pi_S : \RR^L \to \RR^{L'}$ is the projection onto the coordinates in $S\subset [L]$, with $|S|=L'$. Here, the set $S$ is fixed across all samples, and is a priori known.

As before, we are interested in asymptotics as $L,L',\sigma^2\to \infty$ simultaneously. In the PMRA, 
{we paramterize the noise as}
$\sigma^2=\frac{L'}{\alpha\log(L)}${; this is smaller than how we scaled $\sigma^2$ in MRA by a factor of $L'/L$}. The numerator $L'$ comes from the total ``signal energy'' that each measurement sees: $\Expt \|\pi_S R_{\ell_i}X\|^2=L'$, whereas the $\log(L)$ factor is $\log$ the size of the group of shifts (and therefore is the same as in MRA).

In the interest of space, we only provide a brief sketch for the proof of Theorem~\ref{thm:pmra}. We essentially follow the steps of the proof of Theorem~\ref{thm:highsnrregime}, outlining what modifications need to be made for the argument to work for the PMRA model.

\paragraph{Template matching} The MAP estimator is given by 
\[
\widehat{R}_{\mathrm{MAP}} = \argmax_{\ell'} \frac{\langle \pi_S R_{\ell}'X, Y \rangle }{L'} = \argmax_{\ell'} \left\{ \frac{\langle \pi_S R_{\ell'}X, \pi_S R_\ell X \rangle }{L'} + \frac{\langle \pi_S R_{\ell'}X, \sigma Z \rangle }{L'}\right\} .
\]
One can prove, as in Lemma~\ref{lem:signal-is-nice}, that with high probability
\[
\max_{\ell,\ell'}\left| \frac{\langle \pi_S R_{\ell'} X, \pi_S R_\ell X\rangle }{L'} - \Ind_{\{\ell=\ell'\}}\right| = o(1)
\]
holds. Note that the assumption that $L$ is not too large  with respect to $L'$ (strictly less than exponential in $L'$) is \emph{essential} here: following the proof of Lemma~\ref{lem:signal-is-nice}, we can obtain a concentration bound of the form 
\[
\Pr\left( \left|(L')^{-1}\langle \pi_S R_{\ell'} X, \pi_S R_\ell X\rangle- \Ind_{\{\ell=\ell'\}}\right|>\kappa \right) \le \exp(-cL'\min(\kappa,\kappa^2)),
\]
which needs to beat a union bound over all indices $\ell,\ell'$. Having shown that, we can compare the maximum of the noise term to the maximum of a sequence of standard Gaussians (using Lemmas~ \ref{lem:gaussain-max-exp}, \ref{lem:sudakov-fernique} and \ref{lem:borell-tis}), to deduce
\[
\max_{\ell'=0,\ldots,L-1} \frac{\langle \pi_S R_{\ell'}X, \sigma Z \rangle }{L'} \approx \frac{1}{\sqrt{\alpha \log(L)}} \max_{\ell'=0,\ldots,L-1} \left\langle \frac{\pi_S R_{\ell'}X}{{\|\pi_S R_{\ell'}X\|}} , Z\right\rangle \approx \sqrt{\frac{2}{\alpha}} .
\]
Since $\frac{\langle \pi_S R_{\ell}X, \pi_S R_\ell X \rangle }{L'} \approx 1$, we conclude that the MAP estimator is successful consistently when $\alpha>2$ and fails consistently when $\alpha<2$.

\paragraph{Lower bound at high SNR ($\alpha>2$)} The lower bound on the sample complexity follows from {applying Corollary~\ref{prop:IT-lower-bound} with the following easy bound on the multi-sample MI $I(X;Y^n)$:
	\begin{align}
		I(X;Y^n) \le \frac{L}{2}\log\left( 1 + \frac{L'}{L}n\sigma^{-2}\right) .\label{eq:singleMI_PMRA}
	\end{align}
	The idea for proving~\eqref{eq:singleMI_PMRA} is as follows.} Suppose that the shifts $R_{\ell_1},\ldots,R_{\ell_n}$ were all known. Each measurement $Y_i$ contains noisy measurements of $L'$ out of $L$ coordinates of $X$, and note that if we knew the shifts, we would also know to which coordinate of $X$ each coordinate of $Y_i$ corresponds. For each coordinate $i\in [L]$, let $n_i$, be the total number of (noisy) measurements of $X_i$ available across all samples $Y_1,\ldots,Y_n$. Thus, assuming the shifts are given and known, we can think of the problem as follows: we have $L$ independent standard (one dimensional) Gaussians, $X_1,\ldots,X_L$; for each $i$, we measure $n_i$ measurements of $X_i$ through an AWGN. Thus,
\begin{align*}
	I(X;Y^n|R^n=r^n) \le \sum_{i=1}^L \frac{1}{2}\log\left(1+n_i(r^n) \sigma^{-2}\right) &\le \frac{L}{2} \log\left(1 + \frac{n_1(r^n)+\ldots+n_L(r^n)}{L}\sigma^{-2}\right)\\&=\frac{L}{2}\log\left(1+ \frac{L'}{L}n\sigma^{-2}\right),
\end{align*}
where the second inequality follows from convexity. Averaging over all possible shifts $r^n$, $I(X;Y^n)\le I(X;Y^n|R^n)\le \frac{L}{2}\log\left(1+\frac{L'}{L}n\sigma^2\right)$, as claimed.

\paragraph{Lower bound at low SNR ($\alpha \le 2$)} We can {reiterate} the Fano-type argument of Proposition~\ref{prop:MI-alpha-le-2} without substantial modifications. { 
	The single-sample MI from equation~\eqref{eq:MI-entropy-gap} now becomes
	\[
	I(X;Y) = \frac{L'}{2}\log(1+\sigma^{-2}) - \log(L) + H(R|X,Y).
	\]
	Lemma~\ref{lem:not-many-candidiates} goes through almost verbatim with
	\begin{equation*}
		\m{S}_{\tau} = \left\{ R'\,:\, \frac{\langle \pi_S R' X,Y\rangle }{\|\pi_S R' X\|^2} \ge 1-\tau \right\} .
	\end{equation*}
	instead of the definition given in~\eqref{eq:S-tau}, and with the first term in the left-hand-side of~\eqref{eq:listsizeUB} decaying exponentially fast in $L'$, rather than $L$. 
	Thus,} by the same argument as in the proof of Proposition~\ref{prop:MI-alpha-le-2}, we bound
\[
H(R|X,Y) \le \left(1-\frac{\alpha}{2}+o(1)\right)\log(L). 
\]
Expanding $\frac{L'}{2}\log(1+\sigma^{-2})=\frac{L'\sigma^{-2}}{2}+O(L'\sigma^{-4})$ and plugging $\sigma^2=L'/(\alpha\log(L))$, we conclude that $I(X;Y)=o(\log(L))$. Combining with Proposition~\ref{prop:singlesampleMIbasedBounds},
\[
\SCp(L,\alpha,\varepsilon)\ge \log\left(1/\varepsilon\right) \cdot \frac{L}{2I(X;Y)}(1+o(1)) = \omega\left(\frac{L}{\log(L)}\log\left(1/\varepsilon\right)\right) = \omega\left(\frac{L}{L'} \sigma^2\log\left(1/\varepsilon\right)\right).
\]

%
%

\section{Proofs of Section~\ref{sec:upper-bound}}
\label{sec:proofs_sec_algorithm}

\subsection{Proof of Lemma~\ref{lem:XinX}}
\label{sect:proof-lem:XinX}

Recall that $\kappa$ was chosen so that the first constraint holds with probability $1-o(1)$.
All that remains, then, is to show that $\|\m{F}^{*} X\|_\infty \le \sqrt{10\log(L)}$ holds with high probability. Let $f_\ell \in \CC^L$ be the $\ell$-th DFT basis vector, so that $(\m{F}^*X)_\ell = \langle X,f_\ell \rangle$. Observe that the real and imaginary parts of $(\m{F}^*X)_\ell$ are both Gaussians, with variances bounded by $1$. Hence,
\begin{align*}
	\Pr\left(|(\m{F}^*X)_\ell|^2 > 10\log(L)\right) 
	&\le \Pr\left(|\Re(\m{F}^*X)_\ell|^2 > 5\log(L)\right) + \Pr\left(|\Im(\m{F}^*X)_\ell|^2 > 5\log(L)\right) \\
	&\le 4e^{-\frac{5}{2}\log(L)} = 4L^{-5/2},
\end{align*}
so $\Pr\left(\|\m{F}^* X\|_{\infty}>\sqrt{10\log(L)}\right)\le L\cdot 4L^{-5/2}=4L^{-3/2}=o(1)$.

\subsection{Proof of Lemma~\ref{lem:not-many-large-shifts}}

Bounding $\Ind[|X|\geq a]\leq \frac{|X|}{a}$, as in the proof of Markov's inequality, we have
\[
N_Q(h) = \sum_{\ell=0}^{L-1}\Ind\left[L^{-1}\left| \langle X,R_{\ell}^{-1} Q\rangle \right|^2 \ge h^2 \right] \le h^{-2}L^{-1}\sum_{\ell=0}^{L-1}\left| \langle X,R_{\ell}^{-1} Q\rangle \right|^2 .
\]
We may write 
\[
L^{-1} \sum_{\ell=0}^{L-1} \langle X, R_{\ell}^{-1}Q\rangle ^2 = Q^\T \left(L^{-1} \sum_{\ell=0}^{L-1} (R_{\ell}X)(R_{\ell}X)^\T \right) Q \le \|\m{M}(X)\| ,
\]
where $\m{M}(X)$ is the operator 
\[
\m{M}(X) = L^{-1}\sum_{\ell=0}^{L-1} (R_{\ell}X)(R_{\ell}X)^\T . 
\]
It is convenient to write $\m{M}(X)$ in terms of the DFT basis $f_0,\ldots,f_{L-1}$
\begin{align*}
	\m{M}(X) 
	&= L^{-1}\sum_{\ell=0}^{L-1} \sum_{k,j=0}^{L-1} e^{\frac{2\pi}{L}\ell(k-j)} \langle X,f_k\rangle \langle X,f_j\rangle^* f_k f_j^* \\
	&= \sum_{k=0}^{L-1} |\langle X,f_k\rangle|^2 f_k f_k^* , 
\end{align*}
\revAdd{where we interchanged the order summation and used $\sum_{\ell=0}^{L-1}e^{\frac{2\pi}{L}\ell(k-j)} = L\cdot \Ind[j=k]$. 
	Thus, we see that the DFT basis diagonalizes $\m{M}(X)$, so that its eigenvalues are exactly the magnitudes of the fourier coefficients of $X$, squared.} In particular, $\|\m{M}(X)\|= \|\m{F}^*X\|_\infty^2 \le 10\log(L)$. 

\subsection{Proof of Lemma~\ref{lem:bad-has-small-score}}

Note that $s_1(Q),\ldots,s_n(Q)$ are i.i.d Bernoulli-distributed. Write 
\begin{align*}
	\Pr(s_i(Q)=1) 
	&= \Pr\left( \exists \ell \,:\,L^{-1/2} \langle Y_i,R_{\ell}^{-1}Q\rangle \ge 1-\frac34 \eta \right) \\
	&=  \Pr\left( \exists \ell \,:\,L^{-1/2} \left[ \langle X,R_{\ell}^{-1}Q\rangle + \sigma\langle Z,R_\ell^{-1}Q \rangle \right] \ge 1-\frac34 \eta \right) \\
	&= \Pr\left( \exists \ell \,:\,L^{-1/2} \langle X,R_{\ell}^{-1}Q\rangle + \frac{\langle Z,R_\ell^{-1}Q \rangle}{\sqrt{\alpha \log(L)}} \ge 1-\frac34 \eta \right) .
\end{align*}
Let 
\[
\m{L}(Q) = \left\{ \ell \,:\, L^{-1/2}\langle X,R_{\ell}^{-1}Q\rangle \ge 1-\sqrt{\frac{2}{\alpha}}-\frac{7\eta}{8}  \right\} 
\]
be the set of shifts for which $L^{-1/2}\langle X,R_{\ell}^{-1}Q\rangle$ is somewhat large. For $S\subset[L]$, set 
\[
p(S) = \Pr\left( \exists \ell \in S\,:\, L^{-1/2} \langle X,R_{\ell}^{-1}Q\rangle + \frac{\langle Z,R_\ell^{-1}Q \rangle}{\sqrt{\alpha \log(L)}} \ge 1-\frac34 \eta\right),
\]
so that $\Pr(s_i(Q)=1)\le p(\m{L}(Q)) + p(\overline{\m{L}{(Q)}})$. Since
\[
\Expt \max_{\ell\in \overline{\m{L}(Q)}}\langle Z,R_{\ell}^{-1}Q\rangle \le \sqrt{2\log|\overline{\m{L}(Q)}|} \le \sqrt{2\log(L)},
\]
(since each $\langle Z,R_{\ell}^{-1}Q\rangle \sim \m{N}(0,1)$; see comment after Lemma~\ref{lem:gaussain-max-exp}), 
we apply Lemma~\ref{lem:borell-tis} to get 
\begin{align*}
	p(\overline{\m{L}{(Q)}})
	&\le \Pr\left( \exists \ell \in \overline{\m{L}{(Q)}} \,:\, \frac{\langle Z,R_\ell^{-1}Q \rangle}{\sqrt{\alpha \log(L)}} \ge \sqrt{2/\alpha} + \eta/8 \right) \\
	&\le \Pr\left( \max_{\ell \in \overline{\m{L}{(Q)}}} \langle Z,R_\ell^{-1}Q \rangle \ge \Expt\left[ \max_{\ell \in \overline{\m{L}{(Q)}}} \langle Z,R_\ell^{-1}Q \rangle \right] +  \frac18\eta \sqrt{\alpha \log(L)} \right) \\
	&\le 2e^{-\frac12 (\eta/8)^2\alpha\log(L)   } 
	= 2L^{-\eta^2\alpha/128} . 
\end{align*}
For the other term,
\begin{align*}
	p(\m{L}(Q)) 
	&\le |\m{L}(Q)| \Pr\left( \frac{\langle Z,R_\ell^{-1}Q \rangle}{\sqrt{\alpha \log(L)}} \ge \eta/4 \right) 
	\le |\m{L}(Q)| \cdot e^{- \frac12 (\eta/4)^2\alpha \log(L)} 
	= |\m{L}(Q)| \cdot L^{-\eta^2\alpha/32} .
\end{align*}
By Lemma~\ref{lem:not-many-large-shifts}, 
\[
|\m{L}(Q)| \le \frac{10\log(L)}{\left( 1-\sqrt{\frac{2}{\alpha}}-\frac{7\eta}{8}\right)^2} \le \frac{640\log(L)}{ \left( 1-\sqrt{\frac{2}{\alpha}}\right)^2 },
\]
where we also used $\eta<1-\sqrt{2/\alpha}$. Combining,
\[
\Pr(s_i(Q)=1) 
\le 2L^{-\eta^2\alpha/128} + \frac{640\log(L)}{ \left( 1-\sqrt{\frac{2}{\alpha}}\right)^2 } L^{-\eta^2\alpha/32} 
\le \left(2+\frac{640}{ \left( 1-\sqrt{\frac{2}{\alpha}}\right)^2 }\right)L^{-\eta^2\alpha/128} ,
\] 
where we used the assumption that $L$ is large enough so that $\log(L)\le L^{3\eta^2\alpha/128}$. 
We use
\[
\Pr\left(\mathrm{Binom}(n_1,p)\ge k\right) = \sum_{t=k}^{n_1} \binom{n_1}{t}p^t(1-p)^{n_1-t} \le p^k \sum_{t=k}^{n_1} \binom{n_1}{t} \le 2^{n_1} p^k .
\]
Since $s(Q)\sim \mathrm{Binom}\left(n_1,\Pr(s_i(Q)=1)\right)$,
\begin{align*}
	\Pr\left(s(Q) \ge n_1/2 \right) \le \left[ 16\left(2+\frac{640}{ \left( 1-\sqrt{\frac{2}{\alpha}}\right)^2 }\right)L^{-\eta^2\alpha/128} \right]^{n_1/2}
\end{align*}
as claimed.

\subsection{Proof of Lemma~\ref{lem:good-has-large-score}}

Let $\ell$ be such that $\langle X,R_{\ell}^{-1}Q\rangle \ge 1-5\eta/8$. Then
\begin{align*}
	\Pr(s_i(Q)=0) 
	&\le \Pr\left( \frac{\langle Z,R_\ell^{-1}Q \rangle}{\sqrt{\alpha \log(L)}} < (5/8-3/4)\eta \right)  
	\le L^{-\eta^2\alpha(5/8-3/4)^2} = L^{-\eta^2\alpha/64} \le 1/4.
\end{align*}
Thus, using Hoeffding's inequality,
\begin{align*}
	\Pr\left( s(Q)<n_1/2 \right) \le \Pr\left(\mathrm{Binom}(n_1,3/4)<n_1/2\right) \le e^{-n_1/32}.
\end{align*}

\subsection{Proof of Lemma~\ref{lem:match-with-partial-template}}

\revDel{To simplify the notation, assume without loss of generality that $\ell=\ell^*=0$.}
By the discussion right before the statement of Lemma~\ref{lem:match-with-partial-template} in the main paper,
	the assumption $L^{-1/2}\langle X,R_{\ell^*}^{-1}Q\rangle = L^{-1/2}\langle X,R_{-\ell^*}Q\rangle\ge 1-\eta$ implies that for any other shift,	
	$\ell'\ne -\ell^*$, we already have $L^{-1/2}\langle X,R_{\ell'} Q\rangle \le 1-(\sqrt{2}-\sqrt{2\eta})^2/2+o(1) = \sqrt{4\eta} - \eta + o(1)$. Now, for any $\tau$, 
	\begin{align*}
		\Pr(\widehat{\ell}\ne \ell-\ell^*) 
		&\le \Pr\left( L^{-1/2} \langle Y,R_{\ell-\ell^*}Q\rangle < \tau \quad \textrm{ or } \quad \exists \ell'\ne \ell-\ell^*\,:\,L^{-1/2}\langle Y,R_{\ell'}Q \rangle \ge \tau \right) \\
		&\le \Pr\left(L^{-1/2} \langle Y,R_{\ell-\ell^*}Q\rangle < \tau \right) + \Pr\left(\max_{\ell'\ne \ell-\ell^*} L^{-1/2}\langle Y,R_{\ell'}Q \rangle \ge \tau\right) . 
	\end{align*}
	Recall that 
	\[
	L^{-1/2}\langle Y,R_{\ell'}Q\rangle = L^{-1/2}\langle R_{\ell} X + \sigma Z,R_{\ell'}Q\rangle = L^{-1/2}\langle X,R_{\ell'-\ell} Q\rangle + L^{-1/2}\sigma \langle Z,R_{\ell'}Q\rangle \,,
	\]
	and that the first term on the right-hand-size, $L^{-1/2}\langle X,R_{\ell'-\ell} Q\rangle$, is $\ge 1-\eta$ when $\ell'=\ell-\ell^*$ and $\le \sqrt{4\eta}-\eta +o(1)$ otherwise. 
	Suppose that $\tau \le  1-\eta $. We may bound
	\[
	\Pr\left(L^{-1/2} \langle Y,R_{\ell-\ell^*}Q\rangle < \tau \right)
	\le \Pr\left(L^{-1/2}\sigma\langle Z,R_{\ell-\ell^*}Q\rangle < \tau - (1-\eta)  \right) \le L^{-\frac12 \alpha(\tau-1+\eta)^2}.  
	\]
	As for the other term, recall that $\Expt[\max_{\ell'}\langle Z,R_{\ell'Q}\rangle] \le \sqrt{2\log L}$ (see discussion right after Lemma~\ref{lem:gaussain-max-exp}). Hence, using Lemma~\ref{lem:borell-tis}, 
	if we assume $\tau\ge \sqrt{4\eta} - \eta +\sqrt{2/\alpha} + o(1)$, we may also bound
		{\footnotesize
		\begin{align*}
			\Pr\left(\max_{\ell'\ne \ell-\ell^*} L^{-1/2}\langle Y,R_{\ell'}Q \rangle \ge \tau\right)
			&\le \Pr\left(\max_{\ell'\ne \ell-\ell^*} L^{-1/2}\sigma \langle Z,R_{\ell}Q \rangle \ge \tau - \left( \sqrt{4\eta} - \eta + o(1) \right)\right) \\ 
			&= \Pr\left(\max_{\ell'\ne \ell-\ell^*} \langle Z,R_{\ell}Q \rangle \ge \sqrt{\alpha \log L}\cdot \left[\tau - \left( \sqrt{4\eta} - \eta + o(1) \right)\right]\right) \\
			&\le 
			\begin{aligned}
				\Pr \left( \max_{\ell'\ne \ell-\ell^*} \langle Z,R_{\ell}Q \rangle - \Expt\left[\max_{\ell'\ne \ell-\ell^*} \langle Z,R_{\ell}Q \rangle\right] \ge 
			 \sqrt{\alpha \log L} \left[\tau - \left( \sqrt{4\eta} - \eta + o(1) \right) - \sqrt{\frac{2}{\alpha}}\right]
				\right)
			\end{aligned}\\
			&\le L^{-\frac12 \alpha \left( \tau - \left( \sqrt{4\eta} - \eta + o(1) \right)-\sqrt{2/\alpha} \right)^2} .
		\end{align*}}
	We would now like to choose $\sqrt{4\eta} - \eta + \sqrt{2/\alpha}< \tau < 1-\eta$, so to maximize 
	\[
	\min\left( |\tau-(1-\eta)|, |\tau-(\sqrt{4\eta}-\eta+\sqrt{2/\alpha})| \right) \,.
	\] 
	Indeed, observe that this interval is non-empty exactly iff $\sqrt{4\eta} < 1-\sqrt{2/\alpha}$. The best $\tau$ is then simply the midpoint, $\tau_* = 1/2-\eta + 1/\sqrt{2\alpha} + \sqrt{\eta}$, which gives 
	\[
	\Pr(\widehat{\ell}\ne \ell-\ell^*) \le 2 L^{-\frac12 \alpha\left(1/2-1/\sqrt{2\alpha} - \sqrt{\eta}\right)^2+o(1)} .
	\]

\subsection{Proof of Proposition~\ref{prop:averaging-guarantee}}

Let $\widehat{Q}\in \Sphere^{L-1}$ be the output of Step 1. Let $\m{V}_1$ be the event that $\max_{\ell}\langle X,R_{\ell}^{-1}\widehat{Q}\rangle \ge 1-\eta$, and call the maximizing shift $\ell^*$. By Proposition~\ref{prop:brute-force-guarantee}, \revAdd{having chosen $n_1$ as in the statement of Proposition~\ref{prop:averaging-guarantee},} $\Pr(\m{V}_1)=1-o(1)$. 

Let $Y_1,\ldots,Y_{n_2}$ be $n_2$ new samples (independent of those used for Step 1), and 
let $\m{I} \subset [n_2]$ be the set of misaligned samples, namely,
$
\m{I} = \left\{ i\in [n_2]\,:\, \widehat{\ell}_i \ne \ell_i-\ell^* \right\} 
$ .
We start by providing a high-probability bound on $|\m{I}|$. Lemma~\ref{lem:match-with-partial-template} tells us that conditioned on $\m{V}_1$, the random variables $\Ind_{\{i\in \m{I}\}}$ are i.i.d Bernoullis with $\Pr(\Ind_{\{i\in \m{I}\}}=1)=p\le2 L^{-\frac12 \alpha\left(1/2-1/\sqrt{2\alpha} - \sqrt{\eta}\right)^2+o(1)}  $ (the exponent being strictly negative by our requirements on $\alpha,\eta$),
thus $|\m{I}|\sim \mathrm{Binom}(n_2,p)$. By Bernstein's inequality \revAdd{(see, e.g, Theorem 2.8.4 in \cite{vershynin2018high})},
\begin{align*}
	\Pr\left( |\m{I}| \ge pn_2 + t \,\big|\, \m{V}_1 \right) 
	\le \exp\left( -\frac{\frac12 t^2}{np(1-p) + \frac13 t} \right) . 
\end{align*}
Note that the right hand side is $o(1)$ whenever $t=t(L)$ is such that $t\to\infty$ and $t=\omega(\sqrt{n_2 p})$ as $n_2,L\to\infty$. 
Thus, there is some $c=c(\alpha,\eta)>0$ such that for $K=L^{-c}n_2$, the event $|\m{I}|\le K$ holds with high probability.  

Let $\widehat{\mu}=\frac{1}{n_2}\sum_{i=1}^{n_2} R_{\widehat{\ell}_i}^{-1}R_{\ell_i}X$ and $\widehat{W}=\frac{1}{n_2}\sum_{i=1}^{n_2} R_{\widehat{\ell}_i}^{-1}Z_i$, so that $\widehat{X}=\widehat{\mu}+ \sigma \widehat{W}$. We decompose the error,
\begin{align*}
	L^{-1/2}\|R_{\ell^*}X-\widehat{X}\|&\le L^{-1/2}\|R_{\ell^*}X-\widehat{\mu}\|+L^{-1/2}\sigma \|\widehat{W}\| \le L^{-1/2}\frac{2|\m{I}|}{n_2}\|X\| + L^{-1/2}\sigma\|\widehat{W}\|\\
	&  = \frac{2|\m{I}|}{n_2}(1+o(1)) +  L^{-1/2}\sigma\|\widehat{W}\|.
\end{align*}
We have already argued that with high probability $\frac{2|\m{I}|}{n_2}=o(1)$; it therefore remains to show that for the appropriate choice of $n_2$, the bound $L^{-1/2}\sigma\|\widehat{W}\|\le \sqrt{\varepsilon}$ holds with probability $1-o(1)$. 

Observe that conditioned on $|\m{I}|\le K$, $R_{\ell^*}^{-1}\widehat{W}$ can be written as 
\[
R_{\ell^*}^{-1}\widehat{W}=\frac{1}{n_2}\sum_{i=1}^{n_2} R_{i}Z_i,
\]
where $R_{i}\ne R_{\ell_i}^{-1}$ for at most $K$ indicies. 
Note that the estimated shifts $R_{\widehat{\ell}_i}$ generally depend on the noise $Z_i$, and therefore we cannot simply conclude that $R_{\widehat{\ell_i}}^{-1}Z_i\sim \m{N}(0,I)$, which would have meant that $R_{\ell^*}^{-1}\widehat{W}\sim \m{N}(0,n_2^{-1}I)$. We need to use a slightly more elaborate argument to overcome this difficulty.

For a subset $S\subset[n]$, $|S|= K$, $S=\{i_1,\ldots,i_K\}$, and shifts 
$\bR = (R_1,\ldots,R_K)$,	define
\[
W(S,\bR) = \frac{1}{n_2}\sum_{i\notin S} R_{\ell_i}^{-1}Z_i + \frac{1}{n_2} \sum_{j=1}^K R_j Z_{i_j}.
\]
Conditioned on the high-probability event $|\m{I}|\le K$, we have 
\[
\|\widehat{W}\|=\|R_{\ell^*}\widehat{W}\| \le \max_{|S|=K,\bR\in [L]^K}\|W(S,\bR)\|,
\]
where the maximization is over all possible subsets $S$ of size $K$ and shifts $R_1,\ldots,R_K$.\footnote{\revAdd{For an upper bound, it clearly suffices to consider $S$ of size \emph{exactly} $K$, even when $|\m{I}|$ is smaller than $K$.}} It is therefore enough to show that $L^{-1/2}\sigma\cdot \max_{|S|=K,\bR\in[L]^K}\|W(S,\bR)\|\le \sqrt{\varepsilon}$ holds with probability $1-o(1)$. 
Since the shifts $R_{\ell_i}$ are independent of the noise $Z_i$, for every fixed $S$ and $\bR$ we have $W(S,\bR)\sim \m{N}(0,n_2^{-1}I)$.
Therefore, by a union bound,
\[
\Pr\left( L^{-1/2}\sigma\cdot \max_{|S|=K,\bR\in[L]^K}\|W(S,\bR)\| > \sqrt{\varepsilon} \right) \le n_2^K L^K \Pr\left( \|G\|^2 > \varepsilon \cdot \frac{L}{\sigma^2} \cdot n_2  \right),
\]
where $\RR^L \ni G\sim \m{N}(0,I)$, hence $\|G\|^2$ is a standard $\chi^2$-distributed random variable with $L$ degrees of freedom, and we bounded $\binom{n_2}{K}\le n_2^K$ for the number of possible choices of $S$. Using the tail bound of \cite[Lemma 1]{laurent2000adaptive}:
\begin{align*}
	\Pr_{G\sim \m{N}(0,I)}\left( \|G\|^2 \ge L + 2\sqrt{Lx} + 2x  \right)\le e^{-x}.
\end{align*} 
Plugging in the above bound any $x_0=x_0(L)$ such that $(n_2L)^K e^{-x_0}=o(1)$, that is, $x_0=\omega(K\log(n_2 L))$, we obtain 
\[
\Pr\left( L^{-1/2}\sigma\cdot \max_{|S|=k,\bR\in[L]^K}\|W(S,\bR)\| > \left( \frac{\sigma^2}{n_2} \left[1+ 2\sqrt{x_0/L} + 2x_0/L \right] \right)^{1/2} \right) = o(1),
\]
hence the condition 
\[
n_2 \ge \frac{\sigma^2}{\varepsilon}\left[1+2\sqrt{x_0/L}+2x_0/L\right] 
\]
suffices. Since $K=L^{-c}n_2$, if moreover $n_2=o(L)$ then $x_0=o(L)$, hence $n_2= \gamma_2 \sigma^2/\varepsilon$ for any $\gamma_2>1$ would suffices for large enough $L$.

\end{document}